\documentclass[nonacm,sigconf]{acmart}

\usepackage{balance}
\def\BibTeX{{\rm B\kern-.05em{\sc i\kern-.025em b}\kern-.08emT\kern-.1667em\lower.7ex\hbox{E}\kern-.125emX}}
    
\usepackage{amsthm}
\usepackage{microtype}

\usepackage[utf8]{inputenc}
\usepackage{amsmath}
\usepackage{amsfonts}
\usepackage{amssymb}
\usepackage[linesnumbered,ruled]{algorithm2e}

\newcommand{\BO}[1]{{ O}\left(#1\right)}

\newcommand{\BTO}[1]{\tilde{ O}\left(#1\right)}

\newcommand{\BT}[1]{{\Theta}\left(#1\right)}

\newcommand{\EXP}[1]{\textnormal{E}\left[#1\right]}
\newcommand{\PR}[1]{\textnormal{Pr}\left[#1\right]}

\newcommand{\D}{\mathcal{D}} 
\newcommand{\Sim}{\mathcal{S}} 
\newcommand{\X}{\mathcal{X}} 

\newtheorem{lemma}{Lemma}
\newtheorem{theorem}{Theorem}
\newtheorem{definition}{Definition}

\newcommand{\ip}[2]{\langle{#1},{#2}\rangle}
\newcommand{\p}{\mathbf{p}}
\newcommand{\q}{\mathbf{q}}
\newcommand{\x}{\mathbf{x}}
\newcommand{\y}{\mathbf{y}}

\renewcommand{\a}{\mathbf{a}}
\newcommand{\norm}[1]{\left\lVert #1 \right\rVert}

\begin{document}

\fancyhead{}

\title{Fair Near Neighbor Search:
Independent Range Sampling in High Dimensions}

\author{Martin Aumüller}
\affiliation{
\institution{IT University of Copenhagen, Denmark}
}
\email{maau@itu.dk}

\author{Rasmus Pagh}
\affiliation{\institution{BARC and IT University of Copenhagen, Denmark}}
\email{pagh@itu.dk}

\author{Francesco Silvestri}
\affiliation{\institution{University of Padova, Italy}}
\email{silvestri@dei.unipd.it}

\begin{abstract}
Similarity search is a fundamental algorithmic primitive, widely used in many computer science disciplines.
There are several variants of the similarity search problem, and one of the most relevant is the $r$-near neighbor ($r$-NN) problem: given a radius $r>0$ and a set of points $S$, construct a data structure that, for any given query point $q$, returns a point $p$ within distance at most $r$ from $q$. 
In this paper, we study the $r$-NN problem in the light of fairness.
We consider fairness in the sense of equal opportunity: all points that are within distance $r$ from the query should have the same probability to be returned. In the low-dimensional case, this problem was first studied by Hu, Qiao, and Tao (PODS 2014).
Locality sensitive hashing (LSH), the theoretically strongest approach to similarity search in high dimensions, does not provide such a fairness guarantee. 
To address this, we propose efficient data structures for $r$-NN where all points in $S$ that are near $q$ have the same probability to be selected and returned by the query. 
Specifically, we first propose a black-box approach that, given any LSH scheme, constructs a data structure for uniformly sampling points in the neighborhood of a query.
Then, we develop a data structure for fair similarity search under inner product that requires nearly-linear space and exploits locality sensitive filters. 
The paper concludes with an experimental evaluation that highlights (un)fairness in a recommendation setting on real-world datasets and discusses the inherent unfairness introduced by solving other variants of the problem.
\end{abstract}



\maketitle

\section{Introduction}

In recent years, there has been an increasing interest in building algorithms that achieve \emph{fairness} under certain technical definitions of fairness~\cite{DworkHPRZ12}. This was for example highlighted by the PODS 2019 invited talk on algorithmic fairness given by Venkatasubramanian~\cite{Venkatasubramanian19}. 
The goal is to remove, or at least minimize, unethical behavior such as discrimination and bias in algorithmic decision making.
There is no unique definition of fairness~(see \cite{HardtPNS16} and references therein), but different formulations that depend on the computational problem and on the ethical goals we aim for.
Fairness goals are often defined in the political context of socio-technical systems~\cite{Whi16}, and have to be seen in an interdisciplinary spectrum covering many fields outside computer science~\cite{Selbst19}.
In machine learning, algorithmic fairness is often considered with respect to a classification process.
The concept of ``equal opportunity'' requires that people who can achieve a certain advantaged outcome, such as finishing a university degree or paying back a loan, have equal opportunity of being able to get access to it in the first place, such as getting into a university program or getting approval of a loan. 

Bias can arise even at a low level within basic data structures that are used by decision-making algorithms as a black-box (see discussion in~\cite{HarPeledM19}). 
It is possible that some techniques for boosting performance, like randomization and approximation that result in non-deterministic behavior, add to the overall algorithmic bias. For instance, some database indexes for fast search might give an (unexpected) advantage to some portions of the input data. 

In this paper, we will propose what fairness could mean in the setting of similarity search. 
Similarity search is an important primitive in many applications in computer science such as  machine learning, recommender systems, and data mining. 
One of the most common formulations of similarity search is the $r$-near neighbor ($r$-NN) problem: for a given radius $r>0$, a distance function $\D(\cdot, \cdot)$ that reflects the (dis)similarity between two data points, and a set $S$ of data points, the $r$-NN problem requires to construct a data structure that, given a query point $\q$, returns a point $\p$ such that $\D(\p,\q)\leq r$, if such a point exists. As we will see, common existing data structures for similarity search have a behavior that introduces bias in the output. Our goal is to capture and algorithmically remove this bias from these data structures.


We consider fairness in the sense of equal opportunity. Our goal is to develop a data structure for the $r$-near neighbor problem where all points within distance $r$ from the given query have the same probability to be returned: if $B_S(\q,r)$ is the ball of input points at distance at most $r$ from a query $\q$, we would like that each point in $B_S(\q,r)$ is returned as near neighbor of $\q$ with probability $1/|B_S(\q,r)|$. 
For all constructions presented in this paper, these guarantees hold only in the absence of a
failure event that happens with probability at most $\delta$ for some small $\delta > 0$.
In other words, we aim at solving the following sampling problem:
\begin{definition}\label{def:nns}
Consider a set $S\subseteq \X$ of $n$ points in a metric space $(\X, \D)$. 
The \emph{$r$-near neighbor sampling problem} ($r$-NNS) asks to construct a data structure for $S$ to solve the following task with probability at least $1 - \delta$: Given query $\q$, return a point $\p$ uniformly sampled from the set $B_S(\q,r)$.
\end{definition}

To see an example application for such a system, consider a recommender system used by a newspaper to recommend articles to users. 
Popular recommender systems based on matrix factorization~\cite{KorenBV09} give recommendations by computing the inner product similarity of a user feature vector with all item feature vectors using some efficient similarity search algorithm. 
It is common practice to recommend those items that have the largest inner product with the user. 
However, in general it is not clear that it is desirable to recommend the ``closest'' articles.
Indeed, it might be desirable to recommend articles that are on the same topic but are not \emph{too} aligned with the user feature vector, and may provide a different perspective~\cite{Abiteboul17}.
Knowing a solution to the $r$-NNS problem allows to make a recommendation slightly further away from the user feature vector, but still within a certain distance threshold, as likely to be returned as the closest feature vectors.

The $r$-NNS problem can also be seen as a special case of query sampling in database systems~\cite{Olken1995}, where the goal is to return a random sample of the output of a given query, for efficiently providing statistics on the query.
This could for example be used in discrimination discovery in existing data bases~\cite{LuongRT11}: by performing independent queries to obtain a sample with statistical significance, we can reason about the distribution of attribute types. We could report on discrimination  if the population counts grouped by a certain attribute differ much more than we would expect them to. 
We will return to the implications in database systems in the related work section.

To the best of our knowledge, previous results focused on exact near neighbor sampling in the Euclidean space up to three dimensions \cite{AfshaniP19,AfshaniW17,HuQT14,Olken1995}.
Although these results might be extended to  $\mathbb{R}^d$ for any $d>1$,  they suffer the \emph{curse of dimensionality} as the  query time increases exponentially with the dimension, making the data structures too expensive in high dimensions.
These bounds are unlikely to be significantly improved since several conditional lower bounds show that an exponential dependency on $d$  in query time or space is unavoidable for \emph{exact} near neighbor search (see e.g., \cite{AlmanR15,Williams05}). 

A common solution to the near neigbor problem in high dimensions is provided by the locality-sensitive hashing (LSH) framework proposed by Indyk and Motwani~\cite{IndykM98}. 
In this framework, which is formally introduced in Section~\ref{sec:prelims}, data points are hashed into buckets and only colliding points are inspected when answering a query. Locality-sensitive hash functions are designed in such a way that the collision probability between two points is a decreasing function of their distance~\cite{Charikar02}. As we will show in Section~\ref{sec:prelims}, the standard LSH approach is not suitable for solving the $r$-NNS problem. 
While the uniformity property required in $r$-NNS can be trivially achieved by finding \emph{all} $r$-near neighbor of a query and then randomly selecting one of them, this is computationally inefficient since the query time is a function of the size of the neighborhood. 
One contribution in this paper is the description of a much more efficient data structure, that still uses LSH in a black-box way.

Observe that the definition above does not require different query results to be independent.
If the query algorithm is deterministic and randomness is only used in the construction of the data structure, the returned near neighbor of a query will always be the same.
Furthermore, the result of a query $\q$ might influence the result of a different query $\q'$. 
This motivates us to extend the $r$-NNS problem to the scenario where we aim at independence.

\begin{definition}\label{def:nnis}
Consider a set $S\subseteq \X$ of $n$ points in a metric space $(\X, \D)$.
The \emph{$r$-near neighbor independent sampling problem} ($r$-NNIS) asks to construct a data structure for~$S$ that for any sequence of up to $n$ queries $\q_1, \q_2,\ldots, \q_n$ satisfies the following properties with probability at least $1 - \delta$:
\begin{enumerate}
\item For each query $\q_i$, it returns a point $\textnormal{OUT}_{i,\q_i}$ uniformly sampled from $B(\q_i,r)$;
\item The point returned for query $\q_i$, with $i>1$, is independent of previous query results. That is, for any $\p\in B(\q_i,r)$ and any sequence $\p_1,\ldots,\p_{i-1}$, we have that 
\begin{equation*}
\Pr[ \textnormal{OUT}_{i,\q} {=} \p \mid \textnormal{OUT}_{i-1,\q_{i - 1}} {=} \p_{i - 1}, \ldots, \textnormal{OUT}_{1,\q_1} {=} \p_1] = \frac{1}{|B_S(\q,r)}.
\end{equation*}
\end{enumerate}
\end{definition}
We note that in the low-dimensional setting~\cite{HuQT14,AfshaniW17,AfshaniP19}, the $r$-near neighbor independent sampling problem is usually called \emph{independent range sampling} (IRS).

\subsection{Results}
We propose several solutions to the $r$-NN(I)S problem under different settings. We hope that our fairness view on independent range sampling gives a new perspective to fair similarity search. On the technical side, the paper contributes the following methods:
\begin{itemize}
    \item Section~\ref{sec:sampling} describes a solution to the $r$-near neighbor sampling problem with running time guarantees matching those of standard LSH up to polylog factors. The data structure uses an independent permutation of the data points and inspects buckets according to the order of points under this permutation.
    \item Section~\ref{sec:independent:sampling} shows how to solve the independent sampling case. The query time still matches that of standard LSH up to poly-logarithmic factors. Each bucket is equipped with a count-sketch and the algorithm works by repeatedly sampling points within a certain window from the permutation. 
    \item In Section~\ref{sec:tableau} we introduce an easy-to-implement nearly-linear space data structure  based on the LSH filter approach put forward in~\cite{andoni2017optimal,christiani2017framework}. 
 As each input point appears once in the data structure,  the data structure can be easily adapted to solve the NNIS problem. 
     \item Lastly, in Section~\ref{sec:evaluation} we present an empirical evaluation of (un)fairness in traditional recommendation systems on real-world datasets, and we then analyze the additional computational cost for guaranteeing fairness.
\end{itemize}

We observe that both data structures in Sections~\ref{sec:independent:sampling} and~\ref{sec:tableau} solve the NNIS problem. 
The first data structure is distance-agnostic since it uses as a black-box a given LSH schema and can be automatically applied to any distance measure by just selecting the proper LSH schema.
On the other hand, the second data structure is simpler, uses only near-linear space, but comes with a slower query time. Additionally, it does not use LSH as a black-box and works only for some distances: 
we describe it for similarity search under inner product, although it can be adapted to some other metrics (like Euclidean and Hamming distances) with standard techniques.

\subsection{Previous work}
\label{sec:previous:work}

\paragraph{Independent Range Sampling} 
The IRS problem is a variant of a more general problem known as \emph{random sampling queries}, introduced in the 1980s for computing statistics from a database: given a database and a 
query (e.g., range, relational operator), return a random sample of the query result set. 
The random sample can be used for estimating aggregate queries (e.g., sum or count) or for query optimization. We refer to the survey~\cite{Olken1995Survey} for a more detailed overview.  
One of the first results on range sampling is by Olken and Rotem~\cite{Olken1995}, who proposed a data structure based on a R-tree for sampling points within a region described by a union of polygons. 
Hu et al. introduced in~\cite{HuQT14} the
IRS problem, where the focus is on extracting \emph{independent} random samples.
The paper presents a data structure for the unweighted, dynamic version in one dimension.
Later, Afshani and Wei~\cite{AfshaniW17} proposed a data structure that solves the weighted case in one dimension, and the unweighted case in three dimensions for half-space queries. 
In a very recent paper, Afshani and Phillips~\cite{AfshaniP19} extended this line of work and provide a lower bound for the worst-case query time with nearly-linear space.

\paragraph{Applications of Independent Range Sampling}
Since near-neighbor search in the context of recommender systems works usually on medium- to high-dimensional data sets, we see IRS as a natural primitive to introduce a concept of fairness into recommender systems.
As mentioned in~\cite{AfshaniP19}, IRS is a useful statistical tool in data analysis and has been used in the database community for many years. Instead of obtaining the set of all $r$-near neighbors, which could be a very costly operation, an analyst (or algorithm) might require only a small sample of ``typical'' data points sampled independently from a range to provide statistical properties of the data set.
Another useful application is diversity maximization in a recommender system context. As described by Adomavicius and Kwon in~\cite{adomavicius2014optimization}, recommendations can be made more diverse by sampling $k$ items from a larger top-$\ell$ list of recommendations at random. Our data structures could replace the final near neighbor search routine employed in such systems.  
Finally, we observe that IRS can have applications even in the context of discrimination discovery. For instance Luong et al.~\cite{LuongRT11} used $k$ nearest neighbor search for detecting  significant differences of treatment among users with similar, legally admissible, characteristics. Our data structure can be used in this context for speed up the procedure by sampling a subset of points.

\paragraph{Concurrent Work.}

Concurrently and independently of a first draft of our work\footnote{Note to the reviewer: We omit a reference to preserve anonymity.}, Har-Peled and Mahabadi~\cite{HarPeledM19} considered two variants of the fairness notion discussed in this paper:
\begin{itemize}
	\item {\bf Exact neighborhood.} This coincides with our fairness notion in Definition~\ref{def:nnis} except that sampling probabilities are allowed to differ by a multiplicative factor of $1+\varepsilon$.
	\item {\bf Approximate neighborhood.} This relaxed notion, which allows a faster algorithm, requires sampling to be uniform over some set $S'$ that includes all $r$-near neighbors and does not include any points that are at distance more than $cr$ from~$q$, for some constant $c>1$.
\end{itemize}
From a technical point of view, the algorithm described in~\cite{HarPeledM19} resembles the algorithm defined in Section~\ref{sec:tableau}. The algorithm described there solves the exact neighboorhood version in the special case of
near-linear space LSH data structures. The work of Har-Peled and Mahabadi generalizes this variant to the super-linear space regime by using approximate counting.
For the exact neighborhood variant with fixed $\varepsilon > 0$, their query time bound is $\tilde{O}\left(n^\rho\,\frac{b_S({\bf q},cr)}{b_S({\bf q},r)}\right)$.
Our time bound for $\varepsilon = 0$ is better: $\tilde{O}\left(n^\rho + \frac{b_S({\bf q},cr)}{b_S({\bf q},r)}\right)$, see Theorem~\ref{thm:nns} for details.

For the approximate neighborhood variant with fixed $\varepsilon > 0$, the algorithm of~\cite{HarPeledM19} has query time $\tilde{O}(n^\rho)$, which is better than our query time if $n^\rho \ll \frac{b_S({\bf q},cr)}{b_S({\bf q},r)}$.  
However, we argue that approximate neighborhood is a relatively weak notion of fairness in that it can lead to very different sampling probabilities for different elements at the same distance from ${\bf q}$.
Our discussion uses notation introduced in Section~\ref{sec:nnLSH}, so readers who are interested in the formal details should consult that section.
The fairness guarantee is weak even for the particular sets $S$ that occur in the algorithm of~\cite{HarPeledM19}, namely sets of the form $S' = S({\bf q},cr) \cap \left(\bigcup_i S_{i,\ell_i({\bf q})}\right)$ where $S_{i,\ell_i({\bf q})}$ is defined in~(\ref{eq:bucket}).
With high probability $S({\bf q},r) \subseteq S'$, so each element at distance at most $r$ is sampled with roughly the same probability.
However, for elements at distance between $r$ and $cr$ the sampling probability depends on their local neighborhoods.
Consider elements $x_1, x_2 \in S$ that have the same distance from ${\bf q}$ and each appear in $S'$ with probability $1/2$, while both appear in $S'$ with probability $1/4$.
Furthermore, suppose that conditioned on $x_2\not\in S'$, the expected size of $S'$ is constant.
Then it is easy to see that $x_1$ is sampled with probability $\Omega(1)$.
Moreover, suppose $x_2$ is part of a tight cluster of $\Omega(n)$ points from~$S$ that all have the same hash value with high probability.
Then conditioned on $x_2\in S'$, the size of $S'$ is $\Omega(n)$ with high probability.
This means that $x_2$ is sampled with probability $O(1/n)$.
It is not hard to confirm that the above scenario is possible for concrete LSH families and data sets. We will provide one such example in the experimental evaluation, described in Section~\ref{sec:approx_fair}.

Experiments in~\cite{HarPeledM19} consider the MNIST data set, but only report on the uniformly of samples within the set $S'$.
Thus, they do not shed light on whether sampling probabilities differ in practice because of correlations in the choice of $S'$.

\subsection{Discussion of Fairness Assumptions}
In the context of our problem definition we assume---as do many papers on fairness-related topics---an implicit world-view described by Friedler et al. in~\cite{FriedlerSV16} as ``what you see is what you get''. WYSIWYG means that a certain distance
between individuals in the so-called ``construct space'' (the true merit of individuals) is approximately represented by the feature vectors in ``observed space''. As described in their paper, one has to subscribe to this world-view to achieve certain fairness conditions. Moreover, we acknowledge that our problem definition requires to set a threshold parameter $r$ which might be internal to the dataset. This problem occurs frequently in the machine learning community, e.g., when score thresholding is applied to obtain a classification result. Kannan et al. discuss the fairness implications of such threshold approaches in~\cite{Kannan}. 

We want to stress that we do not claim that the $r$-near neighbor independent sampling problem is \emph{the} fairness definition in the context of similar search. 
We think of it as a suitable starting point for discussion, and acknowledge that the application will often motivate the desired fairness property. 
For example, in the case of a recommender system, we might want to consider a weighted case where closer points are more likely to be returned. As discussed in the concurrent work section (and exemplified in the experimental evaluation), a standard LSH approach does not have such guarantees albeit its monotonic collision probability function.  
We leave the weighted case as an interesting direction for future work.



\section{Preliminaries}
\label{sec:prelims}
 \subsection{Near neighbor search}\label{sec:nn}
Let $(\X,\D)$ be a high dimensional metric space over the set $\X$ with distance function $\D(\cdot, \cdot):\X \times \X \rightarrow \mathbb{R}_{\geq 0}$. As an example, we can set $\X=\mathbb{R}^d$ and let $\D(\cdot, \cdot)$ denote the $\ell_p$ norm, with $p\in \{1,2,+\infty\}$. 
Given a set $S\subseteq \X$ of $n$ points in the metric space $(\X,\D)$ and a radius $r>0$, the \emph{$r$-Near Neighbor} ($r$-NN) problem is to construct a data structure that, for any given query point $\q\in \X$, returns a point $\p\in S$ such that $\D(\q,\p)\leq r$, if such a point exists.
In search of efficient solutions to the $r$-NN problem, several works have targeted the approximate version, named \emph{$(c,r)$-Approximate Near Neighbor} ($(c,r)$-ANN) where $c>1$ is the approximation factor: for any given query $\q\in \X$, the data structure can return a point $\p$ with distance $\D(\q,\p)\leq c\cdot r$.
We refer to the survey \cite{AndoniIR18} for an overview on techniques for approximate near neighbor search in high dimensions. We say that two points $\p,\q$ are \emph{$r$-near} if $\D(\p,\q)\leq r$, \emph{$(c,r)$-near} if $r < \D(\p,\q)\leq c\cdot r$, and \emph{far} if $\D(\p,\q)> c\cdot r$ (if $r$ is clear in the context, we will use the term \emph{near} instead of $r$-near). 
We use $B_S(\q,r)$ to denote the points in $S$ within the ball of radius $r$ and center $\q$ (i.e. $B_S(\q,r) = \{p \in S\colon \D(\p,\q)\leq r\}$), and let $b_S(\q,r) = |B_S(\q,r)|$. 

For the sake of notational simplicity, we assume that  evaluating a distance $\D(\cdot,\cdot)$ or a hash function takes $\BO{1}$ time, and that each entry in $\X$ can be stored in $\BO{1}$ memory words. If this is not the case and a point in $\X$ requires $\sigma$ words and $\tau$ time for reading a point,  computing $\D(\cdot,\cdot)$ or a hash function (with $\sigma\leq \tau)$, then it suffices to add the additive term $n\cdot \sigma$ to our space bounds (we store a point once in memory and then refer to it with constant size pointers), and multiply construction and query times by a factor $\tau$  (since the time of our algorithms is almost equivalent to the number of distance computations or hash computations). We assume the length of a memory word to be  $\BT{\log n}$ bits, where $n$ is the input size.

{\bf Comment.}
In some settings, the distance between points is described in terms of similarity rather than distance.
The metric space $\X$ features a similarity measure $\Sim(\cdot,\cdot):\X \times \X \rightarrow \mathbb{R}_{\geq 0}$. 
The $r$-NN  problem under the similarity measure $\Sim$ requires to construct a data structure that, for any given query point $\q\in \X$, returns a point $\p\in S$ such that $\Sim(\q,\p)\geq r$. The $(c,r)$-ANN version is defined equivalently, with $c\in (0,1)$. 
Some examples used in this paper are the inner product  where $\X=\mathbb{R}^d$ and  $\Sim(\cdot, \cdot)$ denotes the inner product, and the Jaccard similarity where $\X$ is the family of all subsets of $\{1, \ldots, d\}$ where $\Sim(\q,\p) = |\q\cap \p| / |\q\cup \p|$ .

 \subsection{Locality Sensitive Hashing}\label{sec:nnLSH}
Locality Sensitive Hashing (LSH) is a common tool for solving the ANN problem and was introduced in \cite{IndykM98}. 
\begin{definition}
A distribution $\mathcal H$ over maps $h\colon \X \rightarrow U$, for a suitable set $U$, is called $(r, c\cdot r, p_1, p_2)$-sensitive if the following holds for any $\x$, $\y\in  \X$:
\begin{itemize}
\item  if $\D(\x, \y) \leq r$, then $\Pr_h[h(\x) = h(\y)] \geq p_1$;
\item  if $\D(\x, \y) > c\cdot r$, then $\Pr_h[h(\x) = h(\y)] \leq p_2$.
\end{itemize}
The distribution $\mathcal H$ is called an LSH family, and has quality $\rho = \rho(\mathcal{H})  = \frac{\log p_1}{\log p_2}$.
\end{definition}
For the sake of simplicity, we assume that $p_2\leq 1/n$: if $p_2>1/n$, then it suffices to create a new LSH family $\mathcal H_K$ obtained by concatenating $K=\BT{\log_{p_2}(1/n)}$ i.i.d. hashing functions from $\mathcal H$. The new family $\mathcal H_K$ is $(r, cr, p_1^K, p_2^K)$-sensitive and $\rho$ does not change.

The standard approach to $(c,r)$-ANN using LSH functions is the following. Let $\ell_1,\ldots, \ell_L$ be $L$ functions randomly and uniformly  selected from $\mathcal H$.
The data structure consists of $L$ hash tables $H_1,\ldots H_L$: each hash table $H_i$ contains the input set $S$  and uses the hash function $\ell_i$ to split the point set into buckets.
For each query $\q$, we iterate over the $L$ hash tables: for any hash function, compute $\ell_i(\q)$ and compute, using $H_i$, the set 
\begin{equation}\label{eq:bucket}
S_{i,\ell_i(\q)}=\{\p: \p\in S, \ell_i(\p)=\ell_i(\q)\}
\end{equation}
of points in $S$ with the same hash value; then, compute the distance $\D(\q,\p)$ for each point $\p\in S_{i,\ell_i(\q)}$.
The procedure stops as soon as a $(c,r)$-near point is found. It stops and returns $\perp$  if there are no remaining points to check or if it found more than $3L$ far points~\cite{IndykM98}.
By setting $L=\BT{p_1^{-1} \log n}$, the above data structure returns, with probability at least $1-1/n$, a $(c,r)$-near point of $\q$ in time $\BO{L} = \BO{n^\rho \log n}$ and space $\BO{L n}=\BO{n^{1+\rho}\log n}$, assuming that the hash function can be computed in $\BO{1}$ time and each hash value requires $\BO{1}$ space.
By avoiding stopping after $3L$ far points have been found, the above data structure can be adapted to return a $r$-near point of $\q$ in expected time $\BO{n^\rho \log n + b_S(\q, cr)}$.

Standard LSH families do not satisfy the fairness definitions introduced above.
Consider the simple case with the data set $S=\{\x,\y\}$ and where $\D(\x,\y)=r$ and query $\q=\x$: with a standard LSH approach, we have that $\q$ collides with $\x$ with probability 1, while $\q$ collides with $\y$ only once in expectation; therefore, it is likely that the first near point of $\q$ found by the data structure is $\x$. This is true even if the order in which the $L=\BTO{p_1^{-1}}$ hash tables are visited is randomized.
However, if performance is not sought, the standard LSH method can be easily adapted to randomly return a point in $B_S(q,r)$: it suffices to find all near neighbors of $\q$ and to randomly select one of them. However, this approach could require expected time $\BTO{b_S(\q, r) n^\rho +b_S(\q,cr)}$.

\subsection{Sketch for distinct elements}\label{sec:sketch}
In Section \ref{sec:independent:sampling} we will use sketches for estimating the number of distinct elements.
Consider a stream of $n$ elements $x_1,\ldots x_m$ in the domain $[n] = \{1,\ldots, n\}$ and let  $F_0$ be the number of distinct elements in the stream (i.e., the zeroth-frequency moment).
Several papers have studied sketches (i.e., compact data structures) for estimating $F_0$.
For the sake of simplicity we use the simple sketch in \cite{BarYossefJKST02}, which generalizes the seminal result by Flajolet and Martin \cite{FlajoletM85}.
The data structure consists of $\Delta=\BT{\log(1/\delta)}$ lists $L_1,\ldots L_\Delta$; for $1 \leq  w \leq \Delta$, $L_w$ contains the $t = \BT{1/\epsilon^2}$  distinct smallest values of the set $\{\psi_w(x_i) : 1\leq i \leq m\}$, where $\psi_w : [n] \rightarrow [n^3]$ is a hash function picked from a pairwise independent family. 
It is shown in \cite{BarYossefJKST02} that the median  $\hat F_0$  of
the values $tn^3/v_0, \ldots, tn^3/v_\Delta$, where $v_w$ denotes the $t$-th smallest value in $L_w$, is an $\epsilon$-approximation to the number of distinct elements in the stream with probability at least
$1-\delta$: that is,  $(1-\epsilon) F_0 \leq \hat F_0 \leq (1+\epsilon) F_0$.
The data structure requires $\BO{\epsilon^{-2} \log m \log (1/\delta)}$ bits and $\BO{\log(1/\epsilon) \log m \log (1/\delta)}$ query time. 
A nice property of this sketch is that if we split the stream in $k$ segments and we compute the sketch of each segment, then it is possible to reconstruct the sketch of the entire stream 
by combining the sketches of individual segments (the cost is linear in the sketch size).

\section[r-near neighbor sampling]{$r$-near neighbor sampling}
\label{sec:sampling}
We  start with a simple data structure for the $r$-near neighbor sampling problem in high dimensions that leverages on LSH: with high probability, for each given query $\q$, the data structure returns a point uniformly sampled in $B_S(\q,r)$.
We will then see that, with a small change in the query procedure, the data structure supports independent sampling when the same query point is repeated.

The main idea is quite simple. 
We initially assign a (random) rank to each point in $S$ using a random permutation, and then construct the standard LSH data structure for solving  the $r$-NN problem on $S$.
For a given query $\q$ and assuming that all points in $B_S(\q,r)$ collide with $\q$ at least once, the data structure returns the point in $B_S(\q,r)$ with the lowest rank in the random permutation.
The random permutation, which is independent of the collision probability, guarantees that all points in $B_S(\q,r)$ have the same probability to be returned as near point.

\paragraph*{Construction.} Let $\mathcal{H}$ be an $(r,c\cdot r,p_1,p_2)$-sensitive LSH family with $p_2 = \BO{1/n}$ (see Section \ref{sec:nnLSH}). 
Let $\ell_1,\ldots,\ell_{L}$ be $L=\BT{p_1^{-1} \log n}$  hash functions selected independently and uniformly at random from $\mathcal H$.
The $n$ points in $S$ are randomly permuted (possibly, with an universal hash family over maps $\X\rightarrow [n^{3}]$) and let $r(\p)$ denote the rank of point $\p\in S$ after the permutation.
For each $\ell_i$, we partition the input points $S$ according to the hash values of $\ell_i$ and denote with $S_{i,j}=\{\p: \p\in S, \ell_i(\p)=j\}$  the set of points with hash value $j$ under $\ell_i$.
We store points in each $S_{i,j}$ sorted by increasing ranks.
 
\paragraph*{Query.} Let $\q$ be the query point. 
The query procedure extracts from the set $S_\q = \bigcup_{i=1}^{L} S_{i,\ell_i(\q)}$ (i.e., points colliding with $\q$ under the $L$ hash functions) the $r$-near point of $q$ with lowest rank. This is done as follows.
Initialize $r_\text{min} = +\infty$ and $\x_\text{min}=\perp$, where $\perp$ denotes a special symbol meaning "no near neighbor".
For each $i$ in $\{1,\ldots L\}$, scan $S_{i,j}$ 
until an $r$-near point $\x_i$ is found or the end of the array is reached:
in the first case and if $r_\text{min} > r(\x_i)$,  we set $r_\text{min} = r(\x_i)$ and $\x_\text{min}=\x_i$ (we note that $\x_i$ is the $r$-near point in $S_{i,j}$ with lowest rank).
After scanning the $L$ buckets, we return  $\x_\text{min}$.

\begin{theorem}\label{thm:nns}
With high probability $1-1/n$, the above data structure solves the $r$-NNS problem: given a query $\q$, each point $\p\in  B_S(\q,r)$ is returned as near neighbor of $\q$ with probability $1/b_S(\q,r)$.
The data structure requires $\BT{n^{1+\rho}\log n}$ words, $\BT{n^{1+\rho}\log n}$ construction time, and the expected query time is 
\begin{equation*}
\BO{\left(n^{\rho}+\frac{b_S(\q,cr)}{b_S(\q,r) + 1}\right)\log n}.
\end{equation*}
\end{theorem}
\begin{proof}
Since $L = \BT{p_1^{-1} \log n}$, all points in $B_S(\q,r)$ collide with $\q$ at least once (i.e., $B_S(\q,r) \subseteq \bigcup_{i=1}^{L} S_{i,\ell_i(\q)}$) with probability at least $1-1/n$. 
The initial random permutation guarantees that each point in $B_S(\q,r)$ has probability $1/b_S(\q,r)$ to have the smallest rank in $B_S(\q,r)$, and hence to be returned as output.

The space complexity and construction time follow since the algorithm builds and stores $L$ tables with $n$ references to points in $S$.
To upper bound the expected query time, we introduce the random variable $X_{\p,i}$ for each point $\p$ with $\D(\p,\q)>r$: $X_{\p,i}$ takes  value 1 if $\p$ collides with $\q$ under $\ell_i$ and if it has a rank smaller than all points in $B_S(\q,r)$; $X_{\p, i}$ is 0 otherwise. 
Points $\p$ and $\q$ collide with probability $\PR{\ell_i(\q)=\ell_i(\p)}$ and $\p$ has rank smaller than  points in $B_S(\q,r)$ with probability $1/(b_S(\q,r) + 1)$; thus we have that $\EXP{X_{\p,i}} = \PR{\ell_i(\q)=\ell_i(\p)}/(b_S(\q,r) + 1)$ (note that the random bits used for LSH construction and for the initial permutation are independent).
Let $\mu$ be the number of points inspected by the query algorithm over all repetitions. By linearity of expectation, we get:
\begin{align*}
\EXP{\mu} &\leq L + \EXP{\sum_{\p\in S\setminus B_S(\q,r)} \sum_{ i=1}^{L} X_{\p,i}} \\
& \leq   L + L \cdot \sum_{\p\in S\setminus B_S(\q,r)}  \frac{\PR{\ell_1(\q)=\ell_1(\p)}}{b_S(\q,r) + 1}\\
 &\leq  L + L \frac{b_S(\q,cr) p_1 + n p_2}{b_S(\q,r) + 1} = \BO{\left(n^\rho + \frac{b_S(\q,cr)}{b_S(\q,r) + 1}\right)\log n},
\end{align*}
since $p_2 = \BO{1/n}$.
\end{proof}

We observe that our data structure for $r$-NNS automatically gives a data structure for $r$-NN that improves the standard  LSH approach:  the standard approach incurs a $\BTO{n^\rho+b_S(\q,cr)}$  expected query time for worst-case data sets, while our data structure is never worse than $\BTO{n^\rho+b_S(\q,cr)/(b_S(\q,r) + 1)}$ in expectation. 
This is consequence of the  initial random permutation that breaks long chains of consecutive $(c,r)$-near points. We remark that all of our methods have an additional running time term that scales with $b(\q,cr)/b(\q,r)$. This makes running time data-dependent in the following sense: If $c$ is increased (by the user), the $\rho$-value of the LSH will decrease (decreasing the running time), but we pay for it with a possible increase in $b(\q, cr)/b(\q,r)$. However, we will see in Section~\ref{sec:costRatio} that this ratio is small in some real-world datasets.

\subsection[Sampling points with/without replacement]{Sampling points with/without replacement}
The data structure for $r$-NNS can be easily adapted to return $k$ points uniformly sampled \emph{without} replacement from $B_S(\q,r)$: by assuming for well-definedness that $k\leq b_S(\q, r)$, it suffices to return the $k$ points at distance at most $r$ in $S_\q$ with the smallest ranks. 
On the other hand, a set of $k$ points uniformly sampled \emph{with} replacement from $B_S(\q,r)$ can be obtained by slightly changing the above data structure for NNS: we repeat $k$ times the query $\q$, but after each query we perform a random perturbation of ranks. Details of this approach are presented in Appendix~\ref{app:nns_repeated} in the supplemental material. The data structure will sample independently a point in $B_S(\q,r)$ at every repetition of the \emph{same} query $\q$.
Unfortunately, performing a random perturbation of ranks in the NNS data structure does not guarantee independence among different query points, and thus does not solve the NNIS problem. In the following section we extend the data structure in such a way that it guarantees independence.

\section[r-near neighbor independent sampling]{$r$-near neighbor independent sampling}
\label{sec:independent:sampling}

In this section, we present a data structure that solves the $r$-NNIS problem with high probability. 
Let $S$ be the input set of $n$ points and let $\Lambda$ be the sequence of the $n$ input points after a random permutation; the rank of a point in $S$ is its position in $\Lambda$.
We first highlight the main idea of the query procedure, then we describe its technical limitations and how to solve them.

Let $k\geq 1$ be a suitable value that depends on the query point $\q$, and  assume that $\Lambda$ is split into $k$ segments $\Lambda_i$, with $i\in\{0,\ldots, k-1\}$. (We assume for simplicity that $n$ and $k$ are powers of two.) 
Each  segment $\Lambda_i$ contains the $n/k$ points in $\Lambda$ with  rank in $[i\cdot n/ k, (i+1)\cdot n/ k)$ 
We denote with $\lambda_{\q,i}$  the number of near neighbors of $\q$ in $\Lambda_i$, and with $\lambda \geq \max_i\{\lambda_{\q,i}\}$ an upper bound on the number of near neighbors of $\q$ in each segment.
By the initial random permutation, we have that each segment contains at most $\lambda=\BT{(b_S(\q,r)/k)\log n}$ near neighbors with probability at least $1-1/n^2$.
The query algorithm works in three steps: 
A) Select uniformly at random an integer $h$ in $\{0,\ldots, k-1\}$ (i.e., select a segment $\Lambda_h$);
B) With probability $\lambda_{\q,h}/\lambda$ move to step C, otherwise repeat step A;
C) Return a point uniformly sampled among the near neighbors of $q$ in $\Lambda_h$. All random choices are independent.

The above procedure guarantees that the returned near neighbor point is uniformly sampled in $B_S(\q,r)$. 
Indeed, a point $\p\in B_S(\q,r)$ in segment $\Lambda_h$ is sampled and returned at step C with probability 
$\PR{OUT=\p} = \sum_{j=1}^{+\infty} p_j$, where 
$p_j$ is the probability of returning $\p$ 
at the $j$-th iteration (i.e., after repeating $j$ times step A).
We have:
\begin{equation*}
p_j = \left(\sum_{i=0}^{k-1} \frac{1-\lambda_{\q,i}/\lambda}{k}\right)^{j-1} \frac{\lambda_{\q,h}}{k \lambda} \frac{1}{\lambda_{\q,h}} = \left(1-\frac{b_S(\q,r)}{k\lambda}\right)^{j-1} \frac{1}{k \lambda}.
\end{equation*}
The term with exponent $j-1$ is the probability that step A is repeated $j-1$ times; the second term is the probability of selecting segment $\Lambda_h$ during the $j$-th iteration and then to move to step C; the third term is the probability of returning point $\p$ in step C.
Then:
\begin{equation*}
\PR{OUT=\p} = \sum_{j=1}^{+\infty} p_j = \sum_{j=1}^{+\infty}  \left(1-\frac{b_S(\q,r)}{k\lambda}\right)^{j-1} \frac{1}{k \lambda} = 
\frac{1}{b_S(\q,r)}.
\end{equation*}
As all random  choices are taken independently at query time, the solution guarantees the independence among output points required by Definition~\ref{def:nnis}.

The above  approach however cannot be efficiently implemented.
The first problem is that each segment might contain a large number of points with distance larger than $r$: these points are not used by the query algorithm, but still affect the running time as the entire segment must be read to find all near neighbors.
We solve this problem by filtering out far points with LSH: at query time, we use LSH to retrieve only the near neighbors (and a small number of $(c,r)$-near and far points) that are in the selected segment.
A second issue is due to segment size: to improve performance, segments should be small and contain at least one near neighbor of $\q$; thus, the number  $k$ of segments should be set to $b_S(\q,r)$. 
However, $b_S(\q,r)$ is not known at query time.
Hence, we initially set $k=2\hat{s_{\q}}$, where $\hat{s_{\q}}$ is a $1/2$-approximation of the number $s_{\q}$ of points colliding with $\q$; such an estimate can be computed with the count distinct sketch from Section~\ref{sec:prelims}.
As $\hat{s_{\q}}$ can be much larger than $b_S(\q,r)$, we expect to sample several segments without near neighbors: thus, for every $\Sigma=\BT{\log^2 n}$ sampled segments with no near neighbors of $\q$, we repeat the procedure with a smaller value of $k$ (i.e., we set $k=k/2$).
We are now ready to fully describe our data structure. 

\paragraph*{Construction}
Let $\ell_1,\ldots,\ell_{L}$ be $L=\BT{p_1^{-1} \log n}$  hash functions selected independently and uniformly at random from an $(r,cr,p_1,p_2)$-sensitive LSH family $\mathcal H$, with $p_2=\BO{1/n}$. Let $r(\p)$ be the rank of point $\p\in S$ after a random permutation of the $n$ data points. 
For each $\ell_i$, we partition the input points $S$ according to the hash values of $\ell_i$ and denote with $S_{i,j}=\{\p: \p\in S, \ell_i(\p)=j\}$  the set of points with hash value $j$ under $\ell_i$.
We associate to each non-empty bucket $S_{i,j}$: 1) an index (e.g., a balance binary tree) for efficiently retrieving all points in $S_{i,j}$ with ranks in a given range; 2) a count distinct sketch of $S_{i, j}$ (see Section~\ref{sec:sketch}) with $\epsilon=1/2$ and $\delta= 1/(6n^2)$. (For buckets containing less that $\BT{\log n}$ points of $S$, we do not store a count distinct sketch since it requires more space than the points; we generate the sketch from $S_{i,j}$ every time it is required.)
We observe that any segment $\Lambda_i$ can be constructed by collecting all points in every $S_{{i'},j}$ with ranks in $[i\cdot n/ k, (i+1)\cdot n/k)$.

\paragraph*{Query}
Consider a query $\q$ and let $S_\q = \bigcup_{i=1}^{L} S_{i,\ell_i(\q)}$ and $S_{\q,r} = S_\q \cap B(\q, r)$. 
A $1/2$-approximation $\hat{s_\q}$ of $s_\q=|S_\q|$ follows by merging the count distinct sketches associated with buckets $S_{i,\ell_i(\q)}$ for each $i\in \{1,\ldots,L\}$.
We note that $s_{\q,r}=|S_{\q,r}|$ cannot be estimated with sketches since they cannot distinguish near and far points.
 The query algorithm is the following (for simplicity, we assume $n$ to be a power of two):
\begin{enumerate}
    \item Merge all count distinct sketches of buckets $S_{i, \ell_i(\q)}$, for each $i\in \{1,\ldots,L\}$, and compute a $1/2$-approximation $\hat{s_\q}$ of $s_\q$, such that  $s_\q/2\leq \hat{s_\q} \leq 1.5 s_\q$.
	\item Set $k$ to the smallest power of two larger than or equal to $2  \hat{s_\q}$; let $\lambda = \BT{\log n}$, $\sigma_{\textnormal{fail}} = 0$ and $\Sigma=\BT{\log^2 n}$.
\item Repeat the following steps until successful or $k<2$:
\begin{enumerate}
    \item Assume the input sequence $\Lambda$ to be split into $k$ segments $\Lambda_i$ of size $n/k$, where $\Lambda_i$ contains the points in $S$ with ranks in $[i \cdot n/k, (i+1)\cdot n/k)$. 
	\item Randomly select an index $h$ uniformly at random from $\{0,\ldots, k-1\}$. Using the index in each bucket, retrieve the set $\Lambda'_{\q,h} = S_{\q,r} \cap \Lambda_h$, that is all points with rank in $[h \cdot n/k, (h+1)\cdot n/k)$ and with distance at most $r$ from $\q$. 
	Set $\lambda_{\q,h} = |\Lambda'_{\q,h}|$.
    \item Increment $\sigma_{\textnormal{fail}}$.  If $\sigma_{\textnormal{fail}}=\Sigma$, then set $k = k/2$  and $\sigma_{\textnormal{fail}}=0$.
    \item With probability $\lambda_{\q,h}/\lambda$, declare success.
\end{enumerate}
\item If the previous loop ended with success, return a near neighbor of $\q$ sampled uniformly at random in $\Lambda'_{\q,h}$; otherwise return $\perp$.
\end{enumerate}

\begin{theorem}
With probability at least $1-1/n$, the above data structure solves the $r$-near neighbor independent sampling problem.
The data structure requires $\BT{ n^{1+\rho}\log n}$ words, $\BO{  n^{1+\rho}\log^2 n}$ construction time, and the expected query time is:
\begin{equation*}
\BO{\left(n^\rho  +\frac{b_S(\q,c\cdot r)}{b_S(\q,r) + 1}\right)\log^5 n}.
\end{equation*} 
\end{theorem}
\begin{proof}
Let $\q$ be a query and assume that $b_S(\q,r)\geq 1$.
We start by bounding the initial failure probability of the data structure.  
By a union bound, we have that the following three events hold simultaneously with probability at least $1-1/(2n^2)$: 
\begin{enumerate}
\item All near neighbors in $B_S(\q,r)$ collide with $\q$ under at least one LSH function. By suitably setting  the constant in $L=\BT{p_1^{-1} \log n}$, the claim holds with probability at least $1-1/(6n^2)$.
\item Count distinct sketches provide a $1/2$-approximation of $s_\q$. By setting $\delta=1/(6n^2)$ in the count distinct sketch construction (see Section~\ref{sec:sketch}), the approximation guarantee holds with probability at least $1-1/(6n^2)$.
\item When $k=2^{\lceil\log s_{\q,r}\rceil}$,  every segment of size $n/k$ contains no more than $\lambda=\BT{\log n}$ points from $S_{\q,r}$.  As points are initially randomly permuted, the claim holds with probability at least $1-1/(6n^2)$ by suitably setting the constant in $\lambda=\BT{\log n}$. 
\end{enumerate}

From now on assume that these events are true. We proceed to show that, with probability at least $1 - 1/(2n^2)$, the algorithm returns a point uniformly sampled in $B_S(\q,r)$. In other words, for $\p \in B_S(\q, r)$, we show that $\PR{OUT=\p} = 1/b_S(\q,r)$. 

Let us first discuss the additional failure event. There exists a non-zero probability that the algorithm returns $\perp$ even when $b_S(\q,r)\geq 1$: the probability of this event is upper bounded by the probability $p'$ that a near neighbor is  not returned in the $\Sigma$ iteration where $k=2^{\lceil\log s_{\q,r}\rceil}$
(the actual probability is even lower, since a near neighbor can be returned in an iteration where $k>2^{\lceil\log s_{\q,r}\rceil}$).
By suitably setting constants in $\lambda =  \BT{\log n}$ and $\Sigma=\BT{\log^2 n}$, we get:
\begin{equation*}
p' = \left( 1 - \frac{b_S(\q,r)}{k\lambda}\right)^{\Sigma} \leq
e^{-\Sigma b_S(\q,r) / (k\lambda)} \leq e^{\BT{-\Sigma/\log n}} \leq \frac{1}{2n^2}.
\end{equation*}

This shows that with probability at least $1-1/n^2$, the initial three events hold and the algorithm returns a near neighbor.
As each point in $B_S(\q,r)$ has the same probability $1/(k\lambda)$ to be returned during an iteration of step 3, we have that all points in $B_S(\q,r)$ are equally likely to be sampled. 
By applying a union bound, we have that the claim also holds with probability at least $1-1/n$ for any sequence of $n$ queries. Moreover, as soon as the aforementioned events \textsf{1},\textsf{2} and \textsf{3}  hold, the output probabilities are not affected by the random choices used for generating the initial permutation, LSH construction and count distinct sketches. 
Our data structure then solves the $r$-NNIS problem.

We now focus on the space and time complexity of the data structure.
For each bucket $S_{i,j}$, we use $\BO{|S_{i,j}|}$ memory words for storing the index and the count distinct sketch. Therefore, the space complexity of our data structure is dominated by the $L$ LSH tables, that is $\BT{ n^{1+\rho}\log n}$ memory words.
The construction time is $\BT{ n^{1+\rho}\log^2 n}$, where the additional multiplicative $\log n$ factor is due to sketch construction.

The expected query time can be upper bounded by assuming that the algorithm returns a point in $B_S(\q,r)$ only when $k$ is equal to the smallest power of two larger than $b_S(\q,r)$ (i.e., $k=2^{\lceil\log s_{\q,r}\rceil}$).
The cost of each iteration of step 3 is dominated by step 3.b, where the set $\Lambda'_{\q,h}$ is constructed: 
the set is obtained by first extracting all points with rank in $[hn/k, (h+1)n/k)$ from  buckets $S_{i,\ell_i(\q)}$, for each $i\in\{1,\ldots, L\}$, and then by discarding those points with distance larger than $r$ from $\q$.
For each bucket, we expect to find  $b_S(\q,r)/k=\BO{1}$ near neighbors of $\q$,  $b_S(\q,c\cdot r)p_1/k = \BO{b_S(\q,c\cdot r) p_1 /b_S(\q, r)}$ $(c,r)$-near neighbors, and not more than $n p_2 = \BO{1}$ far points.
Since there are $L=\BT{p_1^{-1}\log n}$ buckets and the rank query for reporting points within a given rank costs $\BO{\log n + o}$ in each bucket (where $o$ is the output size), the cost of each iteration of step 3 is $\BO{\left(n^\rho+b_S(\q,c\cdot r)/b_S(\q,r)\right)\log^2 n}$. This bound extends to all $k > 2^{\lceil \log s_{\q, r}\rceil}$, since the denominator in the calculations is only larger in this case.
When $k\sim b_S(\q,r)$, we expect $\BO{\log n}$ iterations of step 3 before  returning a near neighbor of $\q$ (recall that step 3.d is successful with probability $\BO{1/\log n}$).  

Finally, we bound the cost of all iterations where $k>2^{\lceil\log s_{\q,r}\rceil}$. Since $\hat{s_{\q}} \leq 2n$, we observe that the $k$ value is adapted $O(\log n)$ times in step 3.d. For each fixed value 
of $k$, $\BO{\log^2 n}$ iterations of step 3 are carried out. These factors yield the claimed bound on the expected running time. The case $b_S(\q, r) = 0$ is similar to $b_S(\q, r) = 1$ with the algorithm returning $\perp$ after the last iteration.
\end{proof}

\section{Independent sampling in nearly-linear space}
\label{sec:tableau}

In this section we study another approach to obtain a data structure for 
the $r$-NNIS problem. The method uses nearly-linear space by storing each data point exactly once in each of $\Theta(\log n)$ independent data structures.
The presented approach is simpler than the solution found in the previous section since it avoids any additional data structure on top of the basic filtering approach described in~\cite{christiani2017framework}.
It can be seen as a particular instantiation of the more general space-time trade-off data structures that were recently described in~\cite{AndoniLRW17,christiani2017framework}. It can also be seen as a variant of the empirical approach discussed in~\cite{eshghi2008locality} with a theoretical analysis. 
Compared to~\cite{AndoniLRW17,christiani2017framework}, it provides much easier parameterization and a simpler way to make it efficient. We provide here a succinct description of the data structure. All proofs can be found in Appendix~\ref{app:tableau} in the supplemental material.

In this section it will be easier to state bounds on the running time with respect to inner product similarity on unit length vectors in $\mathbb{R}^d$. 
We define the $(\alpha, \beta)$-NN problem analogously to $(c, r)$-NN, replacing the distance thresholds $r$ and $cr$ with $\alpha$ and $\beta$ such that $-1 < \beta < \alpha < 1$. 
This means that the algorithm guarantees that if there exists a point $\p$ with inner product at least $\alpha$ with the query point, 
the data structure returns a point $\p^\ast$ with inner product at least $\beta$ with constant probability. 
The reader is reminded that for unit length vectors we have the relation $\norm{\p - \q}^2_2 = 2 - 2 \ip{\p}{\q}$. We will use the notation $B_S(\q, \alpha) = \{\p \in S \mid \ip{\p}{\q} \geq \alpha\}$ and $b_S(\q, \alpha) = |B_S(\q, \alpha)|$. We define the $\alpha$-NNIS problem analogously to $r$-NNIS with respect to inner product similarity.

\subsection{Description of the data structure}

\paragraph{Construction}
  Given $m \geq 1$ and $\alpha < 1$, 
  let $t = \lceil 1/(1 - \alpha^2)\rceil$ and assume that $m^{1/t}$ is an integer.
  First, choose $tm^{1/t}$ random vectors $\a_{i, j}$, for $i \in [t], j \in [m^{1/t}]$, where each $\a_{i, j} = (a_1, \ldots, a_d) \sim \mathcal{N}(0, 1)^d$ is a vector of $d$ independent and identically  distributed standard normal Gaussians.%
  \footnote{As tradition in the literature, we assume that a memory word suffices for reaching the desired precision. See~\cite{Charikar02} for a discussion.}
Next, consider a point $\p \in S$.
For $i\in [t]$, let $j_i$ denote the index maximizing $\ip{\p}{\a_{i,j_i}}$.
Then we map the index of $\p$ in $S$ to the bucket $(j_1,\ldots,j_t) \in [m^{1/t}]^t$, and use a hash table to keep track of all non-empty buckets. 
Since a reference to each data point in $S$ is stored exactly once, the space usage can be bounded by $O(n + tm^{1/t})$.

\paragraph{Query} Given the query point $\q$, evaluate the dot products with all $t m^{1/t}$ vectors $\a_{i, j}$. 
For $\varepsilon \in (0, 1)$, let $f(\alpha, \varepsilon) = \sqrt{2(1-\alpha^2)\ln(1/\varepsilon)}$.
For $i \in [t]$, let $\Delta_{\q, i}$ be the value of the largest inner product of $\q$ with the vectors $\a_{i, j}$ for $j \in [m^{1/t}]$. 
Furthermore, let $\mathcal{I}_i = \{j \mid \ip{\a_{i,j}}{\q} \geq \alpha \Delta_{\q, i} - f(\alpha, \varepsilon)\}$. 
The query algorithm checks the points in all buckets $(i_1, \ldots, i_t) \in \mathcal{I}_1 \times \dots \times \mathcal{I}_t$, one bucket after the other. 
If a bucket contains a close point, return it, otherwise return $\perp$.

\begin{theorem}
Let $S \subseteq \X$ with $|S| = n$, $-1 < \beta < \alpha < 1$, and let $\varepsilon > 0$ be a constant. 
Let $\rho = \frac{(1-\alpha^2)(1-\beta^2)}{(1 - \alpha\beta)^2}$.
There exists $m = m(\alpha, \beta, n)$ such that
the data structure described above solves the $(\alpha, \beta)$-NN problem with probability at least $1 - \varepsilon$ in linear space and expected time $n^{\rho + o(1)}$.
\end{theorem}

We remark that this result is equivalent to running time statements found in~\cite{christiani2017framework} for the linear space regime, but the method is considerably simpler. The analysis connects storing data points in the list associated with the largest inner product with well-known bounds on the order statistics of a collection of standard normal variables as discussed in~\cite{david2004order}.

\subsection[Solving a-NNIS]{Solving $\alpha$-NNIS}
\label{sec:tableau:independent}

\paragraph*{Construction} Set $L = \Theta(\log n)$ and build $L$ independent data structures $\text{DS}_1, \ldots, \text{DS}_L$ as described above. For each $p \in S$, store a reference from $p$ to the $L$ buckets it is stored in.

\paragraph*{Query} For query $\q$, evaluate all $tm^{1/t}$ filters in each individual DS$_\ell$.
Let $\mathcal{I}$ be the set of buckets $(i_{\ell, 1},\ldots,i_{\ell,t})$ above the query threshold, for each $\ell \in [L]$, and set $T = |\mathcal{I}|$. Enumerate the buckets from 1 to $T$ in a random order over all repetitions.
Let $k_i$ be the number of data points in bucket $i \in [T]$, and set $K= \sum k_i$. 
First, check for the existence of a near neighbor by running the standard query algorithm described above on every individual data structure. This takes expected time $n^{\rho + o(1)} + O\left(\frac{b_S(\q, \beta)}{b_S(\q, \alpha) + 1} \log n\right)$, assuming points in a bucket appear in random order. If no near neighbor exists, output $\perp$ and return. Otherwise, perform the following steps until success is declared:

\begin{enumerate}
\renewcommand{\theenumi}{\Alph{enumi}}
\item Choose a random integer $i$ in $\{1, \ldots, T\}$, where the probability of choosing $i$ equals $k_i/K$.
\item Choose a random point $\p$ from bucket $i$ uniformly at random. Using the $L$ references of $\p$ to its buckets, 
    compute $c_\p$, $0 \leq c_\p \leq L$, the number of times $\p$ occurs in the $T$ buckets, in time $O(L)$.
\item If $\p$ is a near neighbor, report $\p$ and declare success with probability $1/c_\p$.
\item If $\p$ is a far point, remove $\p$ from the bucket and decrement $k_i$ and $K$.
\end{enumerate}
After a point $\p$ has been reported, move all far points removed during the process into their bucket again. We assume that removing and inserting a point takes constant time in expectation. 

\begin{theorem}
    Let $S \subseteq \X$ with $|S| = n$ and $-1 < \beta < \alpha < 1$. The data structure described above solves the $\alpha$-NNIS problem in nearly-linear space and expected time $n^{\rho + o(1)} + O((b_S(\q, \beta)/(b_S(\q, \alpha) + 1)) \log^2 n)$. 
\end{theorem}

\begin{proof}
    Set $L = \Theta(\log n)$ such that with probability at least $1-1/n^2$, all points in $B_S(\q, \alpha)$ are found in the $T$ buckets. 
    Let $\p$ be an arbitrary point in $B_S(\q, \alpha)$. 
    We show that $\p$ is returned by the query algorithm with probability $1/b_S(\q, \alpha)$. 
    This statement follows by the more general observation that the point $\p$ picked in step \textsf{B} is a weighted uniform point among all points present in the buckets after $i$ iterations of the algorithm. That means that if there are $K'$ points in the buckets, the probability of choosing $\p$ is $c_\p/K'$. If $\p$ resides in bucket $i$, the probability of reporting $\p$ is $k_i/K' \cdot 1/k_i \cdot 1/c_\p = 1/(K'c_\p)$. Since $\p$ is stored in $c_\p$ different buckets, the
    probability of reporting $\p$ is $1/K'$. Since this property holds in each round, each near neighbor has equal chance of being reported by the algorithm.   

    We proceed to proving the running time statement in the case that there exists a point in $B_S(\q, \alpha)$. (See the discussion above for the case $b_S(\q, \alpha) = 0$.)
    Observe that evaluating all filters, checking for the existence of a near neighbor, removing far points, and putting far points back into the buckets contributes $n^{\rho + o(1)}$ to the expected running time (see Appendix~\ref{app:tableau} in the supplemental material for details). 
    We did not account for repeatedly carrying out steps \textsf{A}--\textsf{C} yet for rounds in which we choose a non-far point. 
    To this end, we next find a lower bound on the probability that the algorithm declares success in a single such round.
    First, observe that there are $O(b_S(\q, \beta) \log n)$ non-far points in the $T$ buckets (with repetitions).
    Fix a point $\p \in B_S(\q, \alpha)$. 
    With probability $\Omega(c_\p/(b_S(\q, \beta) \log n))$, $\p$ is chosen in step B. 
    Thus, with probability $\Omega(1/(b_S(\q, \beta) \log n))$, success is declared for point $\p$. 
    Summing up probabilities over all points in $B_S(\q, \alpha)$, we find that the probability of declaring success in a single round is $\Omega(b_S(\q, \alpha)/(b_S(\q, \beta) \log n))$. 
    This means that we expect $O(b_S(\q, \beta)\log n/b_S(\q, \alpha))$ rounds until the algorithm declares success.
    Each round takes time $O(\log n)$ for computing $c_\p$ (which can be done by marking all buckets that are enumerated), so we expect to spend time $O((b_S(\q, \beta) / b_S(\q, \alpha))\log^2 n)$ for these iterations, which concludes the proof.
\end{proof}


\section{Experimental Evaluation}\label{sec:evaluation} 

This section presents a principled experimental evaluation that sheds light on the general fairness implications of our problem definitions. 
The evaluation contains both a validation of the (un)fairness
of traditional approaches in a recommendation setting on real-world datasets, an empirical study of unfairness in approximate approaches (such as presented in~\cite{HarPeledM19}), and a short discussion of the additional cost introduced by solving the exact neighborhood problem. The code, raw result files, and the experimental log containing more details are available at \url{https://github.com/alfahaf/fair-near-neighbor-search}. 

We stress that the evaluation is meant to be orthogonal to the data structure questions solved in the previous sections. 
The aim of this evaluation is 
to complement the theoretical study with a case study focusing on the fairness implications of solving the near-neighbor problem. 

To carry out this study, we chose 
the MovieLens and Last.FM datasets available at \url{https://grouplens.org/datasets/hetrec-2011} and converted them to a set representation. 
This representation was obtained by collecting, for each user, all movies rated at least 4 (MovieLens, 2\,112 users, 65\,536 unique movies) and the top-20 artists (Last.FM, 1\,892 users, 18\,739 unique artists), respectively. 
The average set size is 178.1 ($\sigma = 187.5$) and 19.8 ($\sigma = 1.78$), respectively. 
We measure the similarity of two user sets $\mathbf{X}$ and $\mathbf{Y}$ by their Jaccard similarity $J(\mathbf{X},\mathbf{Y}) = |\mathbf{X} \cap \mathbf{Y}| / |\mathbf{X} \cup \mathbf{Y}|$. 
For each dataset, we pick 50 queries randomly from a set of ``interesting''  users. A user $\mathbf{X}$ is \emph{interesting} if there exist at least 40 other users with Jaccard similarity at least $0.2$ with $\mathbf{X}$.

\paragraph{Algorithms.} Since we put our focus on Jaccard similarity, we implemented LSH using standard MinHash~\cite{bro97b} and applying the 1-bit scheme of Li and König~\cite{LiK10}. 
The implementation takes two parameters $K$ and $L$, as discussed in Section~\ref{sec:nnLSH}. 
We set $K$ such that we expect no more than $5$ points with Jaccard similarity at most 0.1 to have the same hash value as the query.
We set $L$ such that with probability at least 99\%, a given point with similarity at least $r \in \{0.15, 0.2, 0.25, 0.3\}$ is present in the $L$ buckets.
The \emph{standard LSH} implementation checks every bucket, stopping as soon as it finds a point with similarity at least $r$. 
We also consider \emph{fair LSH}, which we implemented in the naïve way of collecting all points with similarity at least $r$ found in the buckets, removing duplicates, and returning one of the remaining points at random.

\paragraph{Objectives of the Experiments} Our experiments are tailored to answer the following questions:

\begin{enumerate}
    \item[(Q1)] How (un)fair is the output of standard LSH compared to fair LSH?
    \item[(Q2)] How fair is the output of an algorithm solving the approximate neighborhood version?
    \item[(Q3)] What is the extra cost term for solving the exact neighborhood problem?
\end{enumerate}

\begin{figure*}[t]
    \includegraphics[width=.49\textwidth]{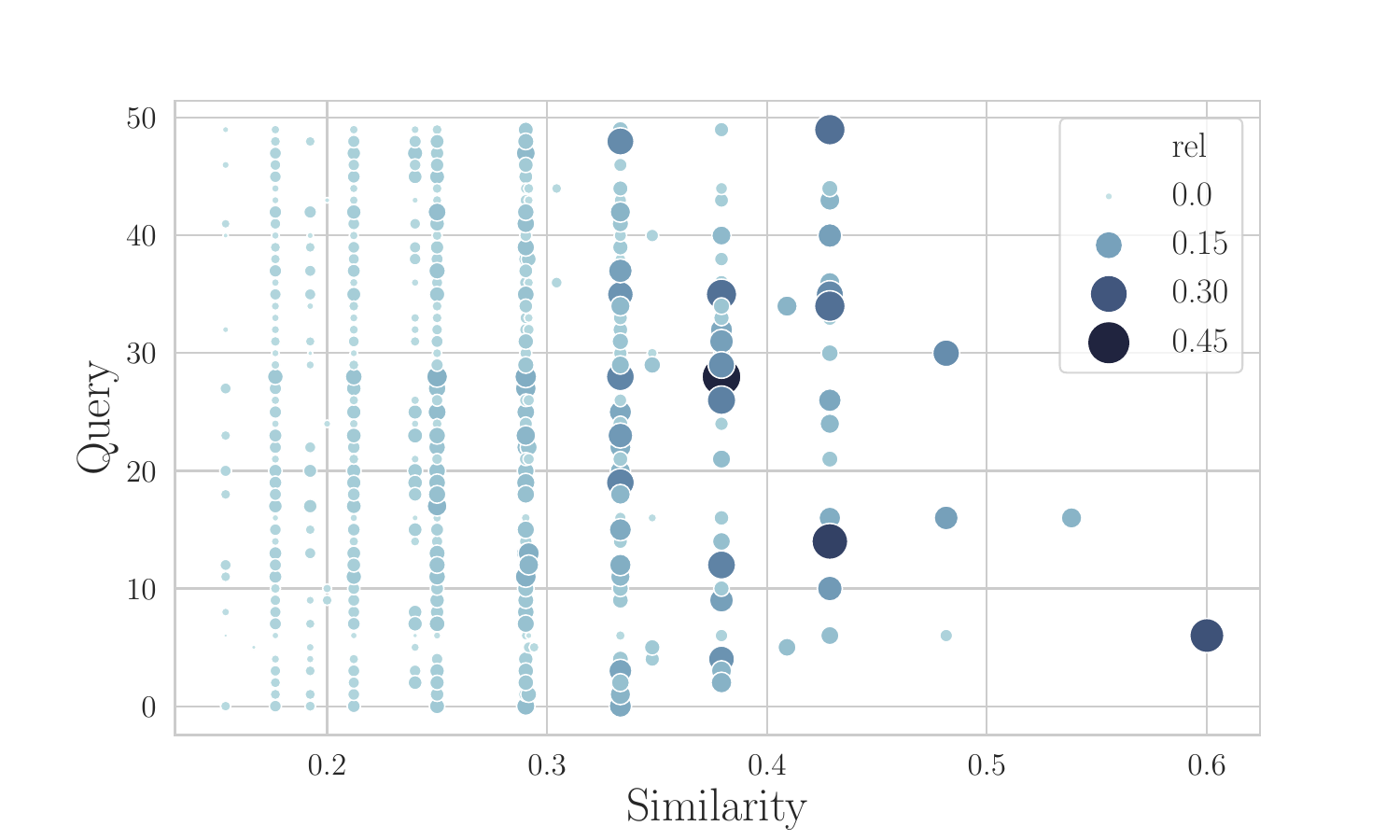}
    \includegraphics[width=.49\textwidth]{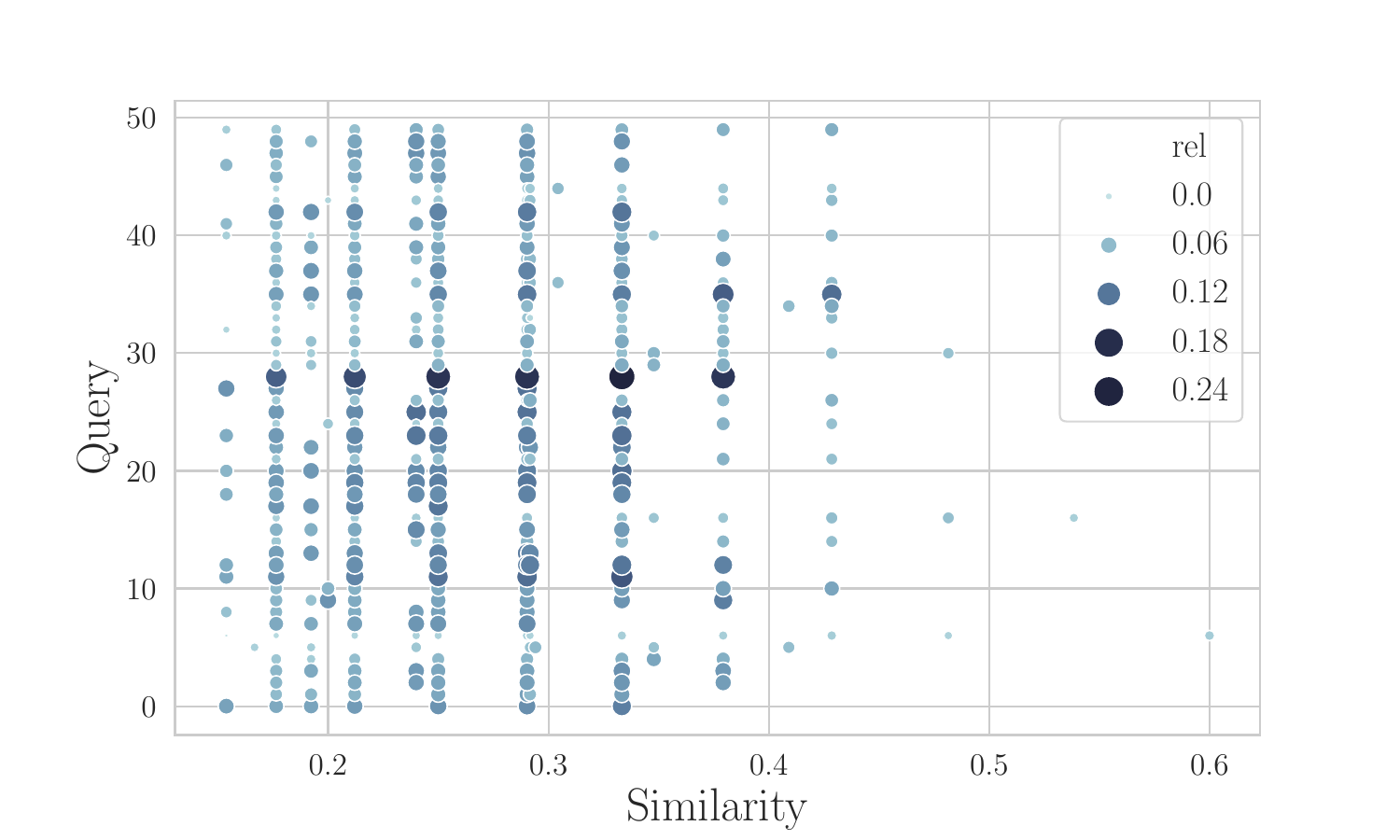}
    \includegraphics[width=.49\textwidth]{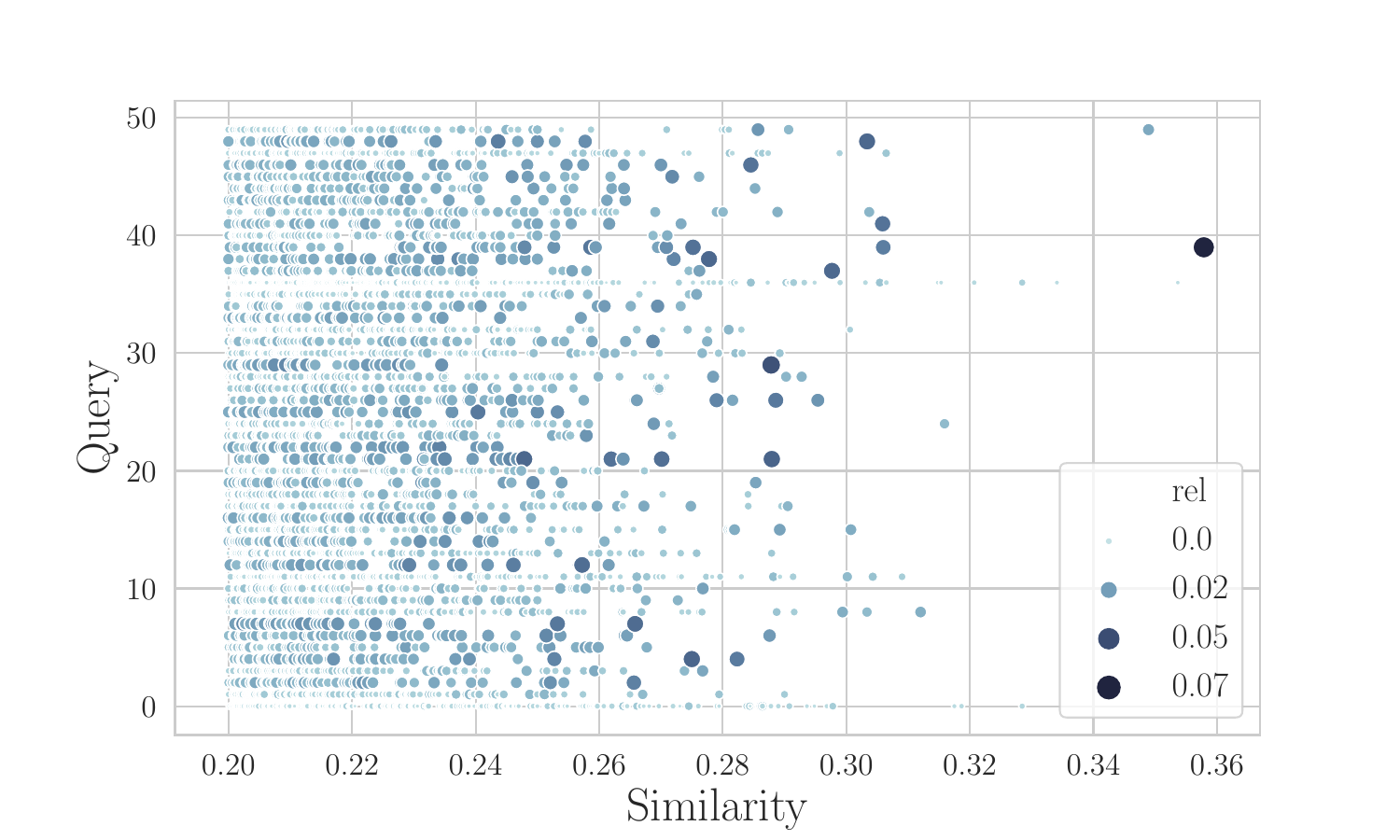}
    \includegraphics[width=.49\textwidth]{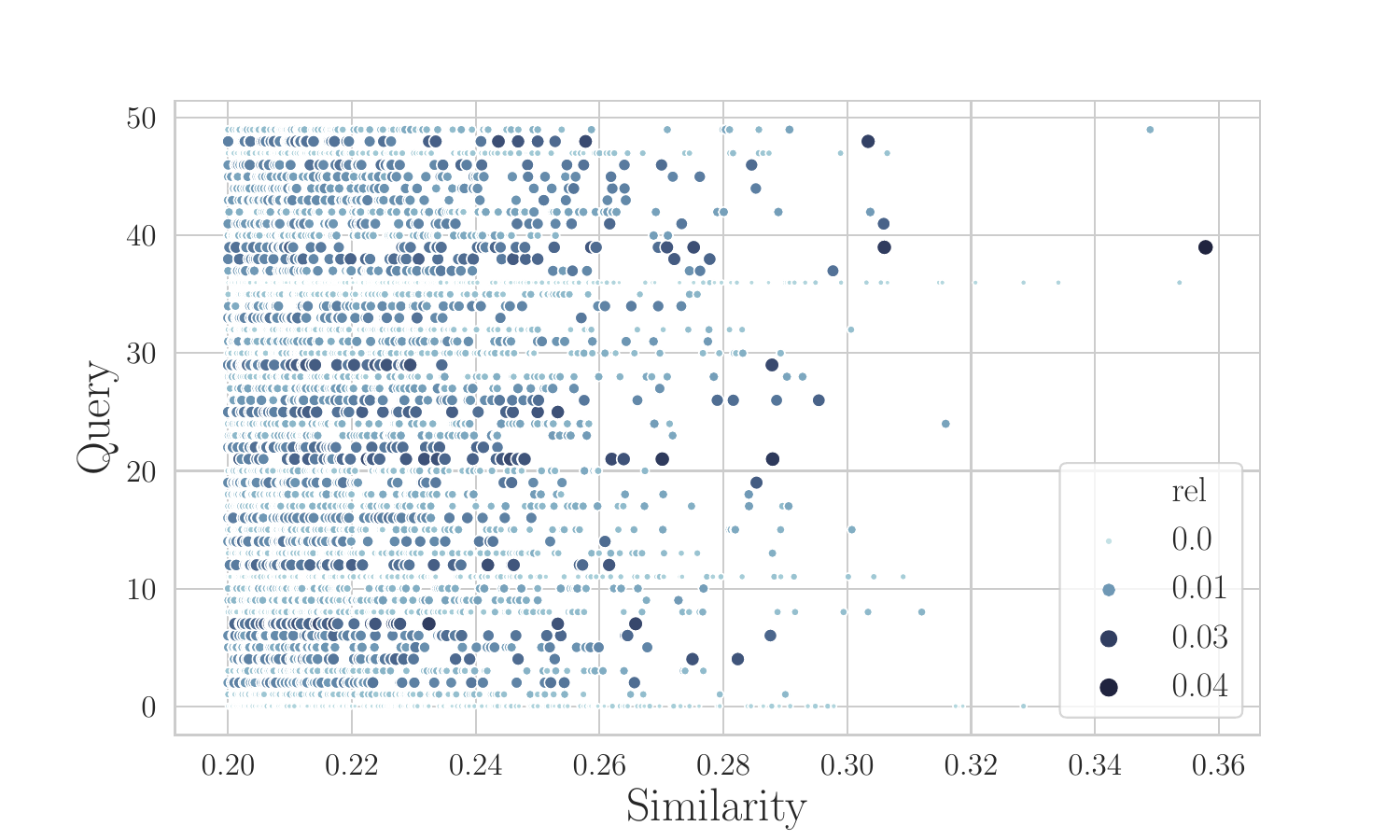}
    \caption{Scatter plot of relative frequencies of reported points; $x$-axis: Similarity to query point; $y$-axis: id of the query. Each point represents the average rel. frequency among all points having this similarity for a fixed query point. Higher color intensity is associated with higher frequency of being reported. Left: standard LSH; right: fair LSH; top: Last.FM ($r = 0.15$), bottom: MovieLens ($r = 0.2$).}
    \label{fig:standardfair}
\end{figure*}

\subsection{Comparing Standard LSH and Fair LSH (Q1)}

For each $r \in \{0.15, 0.2, 0.25, 0.3\}$, we build the standard LSH data structure with the parameters mentioned above. For each of the 50 queries, we query using the standard approach and the fair approach of collecting all points and then returning a neighbor sampled uniformly at random. As result, we collect the reported neighbor. We repeat this process 26\,000 times independently. 

Figure~\ref{fig:standardfair} exemplifies plots of the output distribution obtained from these experiments. The $x$-axis represents the similarity to the query,  
the $y$-axis represents the 50 different queries. 
Each point depicts the relative frequency of a point being reported at the similarity with the given query, where frequencies are averaged over all points that have the same similarity. 
A fair output is produced when for each query, each associated point to this query has approximately the same relative frequency, i.e., the same
size and color intensity in the plot. 
For the Last.FM dataset on the top, we can clearly see that standard LSH (left) is biased towards points with higher similarity. 
For example, the seventh query point has a near neighbor with similarity 0.6 that is clearly overrepresented. 
In contrast, fair LSH provides a uniform output distribution. 
The bottom row shows the output for the MovieLens dataset. Note that the output size is much larger for this dataset and users are much closer with respect to their similarity to the query point. Still, we see a clear gradient in the color intensity for standard LSH (left), meaning that the output is biased towards points closer to the query. Such a gradient is not visible for fair LSH (right). 

From these observations, we can conclude that a straight-forward LSH implementation introduces bias on real-world datasets. Using algorithms that solve the independent sampling version of the $r$-NN problem eliminate such bias.

\subsection{Fairness in the approximate version (Q2)}\label{sec:approx_fair}

We turn our focus to the approximate neighborhood sampling problem as discussed in Section~\ref{sec:previous:work}, which was studied in~\cite{HarPeledM19}. 
In this version of the problem, the algorithm may return points in $B(\q, cr) \setminus B(\q, r)$ as well, which speeds up the query since it avoids additional filtering steps. 
As claimed in Section~\ref{sec:previous:work}, it is not clear whether this notion
provides fairness in the sense of equal opportunity. In the following, we will provide a concrete example that shows that this is not the case. 

To this end, 
let us define the following dataset over the universe $\mathcal{U} = \{1,\ldots,30\}$. 
We let $\mathbf{X} = \{16, \ldots, 30\}$,  $\mathbf{Y} = \{1, \ldots, 18\}$, and $\mathbf{Z} = \{1, \ldots, 27\}$.
Further, we let $\mathcal{M}$ contain all subsets of $\mathbf{Y}$ having at least $15$ elements (excluding $\mathbf{Y}$ itself). 
The dataset is the collection of $\mathbf{X}, \mathbf{Y}, \mathbf{Z}$ and all sets in $\mathcal{M}$. Let the query $\mathbf{Q}$  be the set $\{1, \ldots, 30\}$. To build the data structure, we  
set $r = 0.9$ and $cr = 0.5$. It is easy to see that $\mathbf{Z}$ is the nearest neighbor with similarity 0.9. $\mathbf{Y}$ is the second-nearest point with similarity 0.6, but $\mathbf{X}$ is among the points with lowest similarity of 0.5. Finally, we note that each $\mathbf{x}\in \mathcal{M}$ has similarity ranging from 0.5 to $0.5\overline{6}$.
Under the approximate neighborhood sampling problem, all points can be returned for the given query.

As in the previous subsection, the algorithm collects all the points found in the $L$ buckets and returns a point picked uniformly at random among those points having similarity at least 0.5. 
Figure~\ref{fig:difficult:sampling} shows the sampling probabilities 
of the sets $\mathbf{X}, \mathbf{Y}, \mathbf{Z}$.
The plot clearly shows that the notion of approximate neighborhood does not provide a sensible guarantee on the individual fairness of users. The set $\mathbf{X}$ is more than 50 times more likely than $\mathbf{Y}$ to be returned, even though $\mathbf{Y}$ is more similar to the query. This is due to the clustered neighborhood of $\mathbf{Y}$, making many other points appear at the same time in the buckets. On the other hand, $\mathbf{X}$ has no close points in its neighborhood (except
$\mathbf{Z}$ and $\mathbf{Q}$). 

We remark that we did not observe this influence of clustered neighborhoods on the two real-world datasets. However, it is important to notice that approximate neighborhood \emph{could} introduce unintentional bias
and, furthermore, can be exploited by an adversary to discriminate a given user (e.g., an adversary can create a set of objects $\mathcal{M}$ that obfuscate a given entry $\mathbf{Y}$.)

Based on this example, one could argue that the observed sampling behavior is intentional. 
If the goal is to find good representatives in the neighborhood of the query, then it certainly seems preferable that $\mathbf{X}$ is reported with such high probability (which is roughly the same as all points in $\mathcal{M}$ and $\mathbf{Y}$ combined). 
Our notion of fairness would make the output less diverse, since $\mathbf{X}$ clearly stands out from many other points in the neighborhood, but it is supposed to be treated in the very same way. 
Such a trade-off between diversity and fairness was also observed, for example, by Leonhardt et al.~in~\cite{LeonhardtAK18}.

\begin{figure}[t]
    \includegraphics[width=\columnwidth]{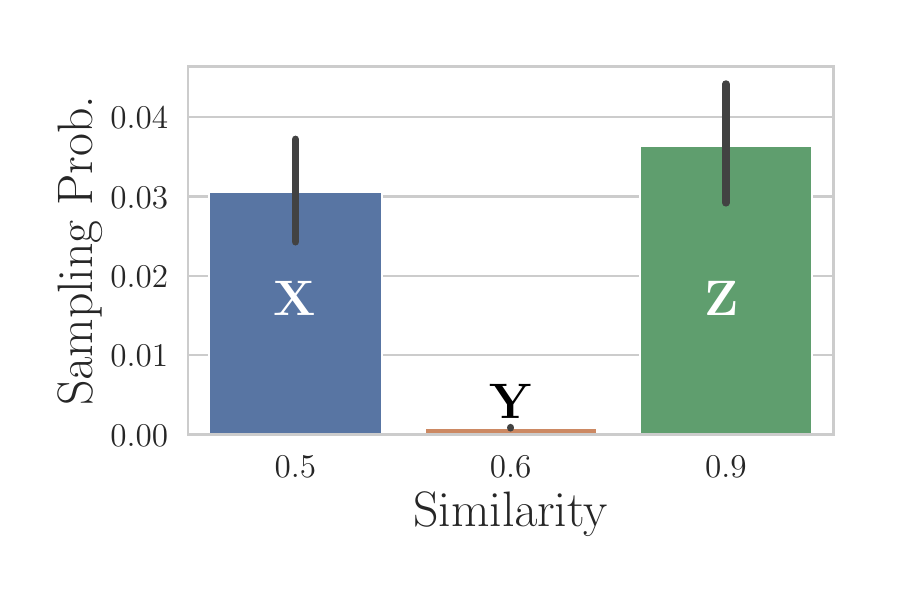}
    \caption{Empirical sampling probabilities of points $\mathbf{X}, \mathbf{Y}, \mathbf{Z}$ with error bars showing 25\% and 75\% quartiles.}
    \label{fig:difficult:sampling}
\end{figure}

\subsection{Additional cost factor for solving the exact neighborhood variant (Q3)}\label{sec:costRatio}
The running time bounds of the algorithms  in Sections~\ref{sec:sampling}--\ref{sec:tableau} have an additional running time factor $\tilde{O}(b_S(\q, cr)/b_S(\q, r))$, putting in relation the number of points at distance at most $cr, c\geq 1,$ (or similarity at least $cr, c \leq 1,$) to those at threshold $r$. The values $r$ and $cr$ are the distance/similarity thresholds picked when building the data structure. 
In general, a larger gap between $r$ and $cr$ makes the $n^\rho$ term in the running times smaller. 
However, the additive cost $\tilde{O}(b_S(\q, cr)/b_S(\q, r))$ can potentially be prohibitively large for worst-case datasets.

Figure~\ref{fig:ratio} depicts the ratio of points with similarity at least $cr$ and $r$ for the two real-world datasets. 
On the Last.FM dataset, we can see that even for very large gaps, the additional cost term is reasonably small. 
(In particular, much smaller than the number of repetitions.) 
On the MovieLens dataset for $r = 0.25$, we see ratios of around 300 for $c \leq 0.25$. 
In this case, it roughly matches the number of repetitions done for this parameter choice. 
In general, however, it shows that a careful inspection of dataset/query set characteristics might be necessary to find good parameter choices that balance these cost terms.

\begin{figure}[t]
    \includegraphics[width=\columnwidth]{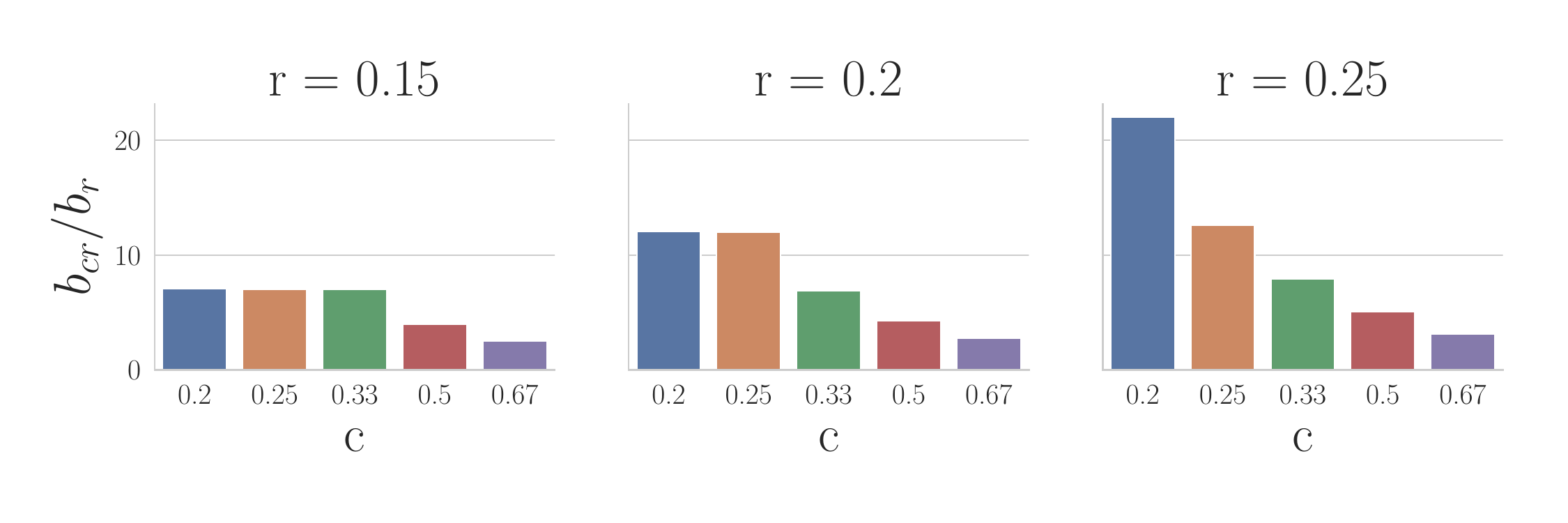}
    \includegraphics[width=\columnwidth]{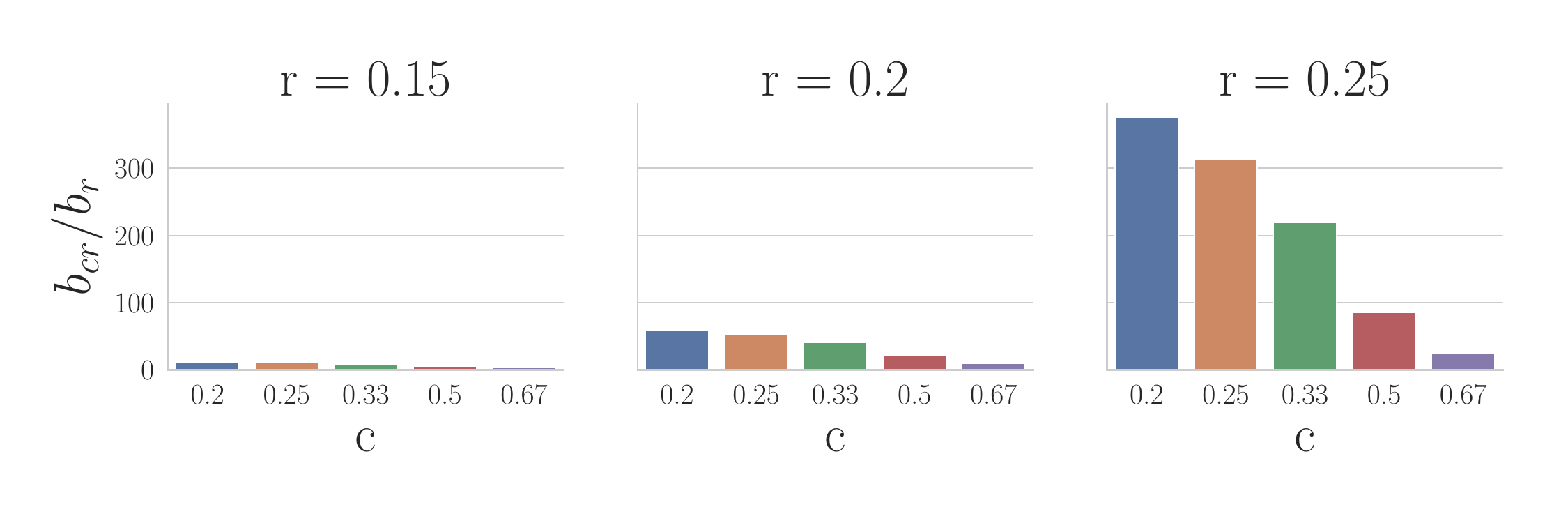}
    \caption{Ratio of number of points with similarity at least $cr$ and number of points with similarity at least $r$ for $r\in\{0.15, 0.2, 0.25\}$ and $c = \{1/5, 1/4, 1/3, 2/3\}$;
    top: Last.FM, bottom: MovieLens.}
    \label{fig:ratio}
\end{figure}

\section{Conclusion}
In this paper, we have investigated a possible definition of fairness in similarity search by connecting the notion of ``equal opportunity'' to independent range sampling.
An interesting open question is to investigate the applicability of our data structures for problems like discrimination discovery \cite{LuongRT11}, diversity in recommender systems~\cite{adomavicius2014optimization},  privacy preserving similarity search \cite{Riazi16}, and estimation of kernel density \cite{CharikarS17}.
Moreover, it would be interesting to investigate techniques for providing incentives (i.e., reverse discriminations~\cite{LuongRT11}) to prevent discrimination: an idea could be to merge the data structures in this paper with  distance-sensitive hashing functions in~\cite{Aumuller18}, which allow to implement hashing schemes where the collision probability is an (almost) arbitrary function of the distance.
Finally, the techniques presented here require a manual trade-off between the performance of the LSH part and the additional running time contribution from finding the near points among the non-far points. From a user point of view, we would much rather prefer a parameterless version of our data structure that finds the best trade-off with small overhead, as discussed in~\cite{AhleAP17} in another setting.

\section*{Acknowledgments}
This work was partially supported by Investigator Grant 16582, Basic Algorithms Research Copenhagen (BARC), from the VILLUM
Foundation,  by the University of Padova under the SID 2018 Project, and  by the Italian Ministry of Education, University
and Research (MIUR), under PRIN Project n. 20174LF3T8 AHeAD.

\bibliographystyle{ACM-Reference-Format}
\bibliography{biblio}


\begin{thebibliography}{38}


\ifx \showCODEN    \undefined \def \showCODEN     #1{\unskip}     \fi
\ifx \showDOI      \undefined \def \showDOI       #1{#1}\fi
\ifx \showISBNx    \undefined \def \showISBNx     #1{\unskip}     \fi
\ifx \showISBNxiii \undefined \def \showISBNxiii  #1{\unskip}     \fi
\ifx \showISSN     \undefined \def \showISSN      #1{\unskip}     \fi
\ifx \showLCCN     \undefined \def \showLCCN      #1{\unskip}     \fi
\ifx \shownote     \undefined \def \shownote      #1{#1}          \fi
\ifx \showarticletitle \undefined \def \showarticletitle #1{#1}   \fi
\ifx \showURL      \undefined \def \showURL       {\relax}        \fi
\providecommand\bibfield[2]{#2}
\providecommand\bibinfo[2]{#2}
\providecommand\natexlab[1]{#1}
\providecommand\showeprint[2][]{arXiv:#2}

\bibitem[\protect\citeauthoryear{Abiteboul, Arenas, Barcel\'{o}, Bienvenu,
  Calvanese, David, Hull, H\"{u}llermeier, Kimelfeld, Libkin, Martens, Milo,
  Murlak, Neven, Ortiz, Schwentick, Stoyanovich, Su, Suciu, Vianu, and
  Yi}{Abiteboul et~al\mbox{.}}{2017}]%
        {Abiteboul17}
\bibfield{author}{\bibinfo{person}{Serge Abiteboul}, \bibinfo{person}{Marcelo
  Arenas}, \bibinfo{person}{Pablo Barcel\'{o}}, \bibinfo{person}{Meghyn
  Bienvenu}, \bibinfo{person}{Diego Calvanese}, \bibinfo{person}{Claire David},
  \bibinfo{person}{Richard Hull}, \bibinfo{person}{Eyke H\"{u}llermeier},
  \bibinfo{person}{Benny Kimelfeld}, \bibinfo{person}{Leonid Libkin},
  \bibinfo{person}{Wim Martens}, \bibinfo{person}{Tova Milo},
  \bibinfo{person}{Filip Murlak}, \bibinfo{person}{Frank Neven},
  \bibinfo{person}{Magdalena Ortiz}, \bibinfo{person}{Thomas Schwentick},
  \bibinfo{person}{Julia Stoyanovich}, \bibinfo{person}{Jianwen Su},
  \bibinfo{person}{Dan Suciu}, \bibinfo{person}{Victor Vianu}, {and}
  \bibinfo{person}{Ke Yi}.} \bibinfo{year}{2017}\natexlab{}.
\newblock \showarticletitle{Research Directions for Principles of Data
  Management (Abridged)}.
\newblock \bibinfo{journal}{\emph{SIGMOD Rec.}} \bibinfo{volume}{45},
  \bibinfo{number}{4} (\bibinfo{year}{2017}), \bibinfo{pages}{5--17}.
\newblock


\bibitem[\protect\citeauthoryear{Adomavicius and Kwon}{Adomavicius and
  Kwon}{2014}]%
        {adomavicius2014optimization}
\bibfield{author}{\bibinfo{person}{Gediminas Adomavicius} {and}
  \bibinfo{person}{YoungOk Kwon}.} \bibinfo{year}{2014}\natexlab{}.
\newblock \showarticletitle{Optimization-based approaches for maximizing
  aggregate recommendation diversity}.
\newblock \bibinfo{journal}{\emph{INFORMS Journal on Computing}}
  \bibinfo{volume}{26}, \bibinfo{number}{2} (\bibinfo{year}{2014}),
  \bibinfo{pages}{351--369}.
\newblock


\bibitem[\protect\citeauthoryear{Afshani and Phillips}{Afshani and
  Phillips}{2019}]%
        {AfshaniP19}
\bibfield{author}{\bibinfo{person}{Peyman Afshani} {and}
  \bibinfo{person}{Jeff~M. Phillips}.} \bibinfo{year}{2019}\natexlab{}.
\newblock \showarticletitle{Independent Range Sampling, Revisited Again}. In
  \bibinfo{booktitle}{\emph{Proc. 35th International Symposium on Computational
  Geometry (SoCG)}}. \bibinfo{pages}{4:1--4:13}.
\newblock


\bibitem[\protect\citeauthoryear{Afshani and Wei}{Afshani and Wei}{2017}]%
        {AfshaniW17}
\bibfield{author}{\bibinfo{person}{Peyman Afshani} {and}
  \bibinfo{person}{Zhewei Wei}.} \bibinfo{year}{2017}\natexlab{}.
\newblock \showarticletitle{Independent Range Sampling, Revisited}. In
  \bibinfo{booktitle}{\emph{Proc. 25th Annual European Symposium on Algorithms
  (ESA)}}. \bibinfo{pages}{3:1--3:14}.
\newblock


\bibitem[\protect\citeauthoryear{Ahle, Aum{\"{u}}ller, and Pagh}{Ahle
  et~al\mbox{.}}{2017}]%
        {AhleAP17}
\bibfield{author}{\bibinfo{person}{Thomas~D. Ahle}, \bibinfo{person}{Martin
  Aum{\"{u}}ller}, {and} \bibinfo{person}{Rasmus Pagh}.}
  \bibinfo{year}{2017}\natexlab{}.
\newblock \showarticletitle{Parameter-free Locality Sensitive Hashing for
  Spherical Range Reporting}. In \bibinfo{booktitle}{\emph{Proc. 28th Symposium
  on Discrete Algorithms (SODA)}}. \bibinfo{pages}{239--256}.
\newblock


\bibitem[\protect\citeauthoryear{Alexandr~Andoni}{Alexandr~Andoni}{2018}]%
        {AndoniIR18}
\bibfield{author}{\bibinfo{person}{Ilya P.~Razenshteyn Alexandr~Andoni,
  Piotr~Indyk}.} \bibinfo{year}{2018}\natexlab{}.
\newblock \showarticletitle{Approximate Nearest Neighbor Search in High
  Dimensions}. In \bibinfo{booktitle}{\emph{Proc. International Congress of
  Mathematicians (ICM)}}. \bibinfo{pages}{3271–3302}.
\newblock


\bibitem[\protect\citeauthoryear{Alman and Williams}{Alman and
  Williams}{2015}]%
        {AlmanR15}
\bibfield{author}{\bibinfo{person}{Josh Alman} {and} \bibinfo{person}{Ryan
  Williams}.} \bibinfo{year}{2015}\natexlab{}.
\newblock \showarticletitle{Probabilistic Polynomials and Hamming Nearest
  Neighbors}. In \bibinfo{booktitle}{\emph{Proc. 56th IEEE Annual Symposium on
  Foundations of Computer Science (FOCS)}}. \bibinfo{pages}{136--150}.
\newblock


\bibitem[\protect\citeauthoryear{Andoni, Laarhoven, Razenshteyn, and
  Waingarten}{Andoni et~al\mbox{.}}{2017a}]%
        {andoni2017optimal}
\bibfield{author}{\bibinfo{person}{Alexandr Andoni}, \bibinfo{person}{Thijs
  Laarhoven}, \bibinfo{person}{Ilya~P. Razenshteyn}, {and}
  \bibinfo{person}{Erik Waingarten}.} \bibinfo{year}{2017}\natexlab{a}.
\newblock \showarticletitle{Optimal Hashing-based Time-Space Trade-offs for
  Approximate Near Neighbors}. In \bibinfo{booktitle}{\emph{Proc. 28th
  Symposium on Discrete Algorithms ({SODA})}}. \bibinfo{pages}{47--66}.
\newblock


\bibitem[\protect\citeauthoryear{Andoni, Laarhoven, Razenshteyn, and
  Waingarten}{Andoni et~al\mbox{.}}{2017b}]%
        {AndoniLRW17}
\bibfield{author}{\bibinfo{person}{Alexandr Andoni}, \bibinfo{person}{Thijs
  Laarhoven}, \bibinfo{person}{Ilya~P. Razenshteyn}, {and}
  \bibinfo{person}{Erik Waingarten}.} \bibinfo{year}{2017}\natexlab{b}.
\newblock \showarticletitle{Optimal Hashing-based Time-Space Trade-offs for
  Approximate Near Neighbors}. In \bibinfo{booktitle}{\emph{Proc. 28th Annual
  {ACM-SIAM} Symposium on Discrete Algorithms (SODA)}}.
  \bibinfo{pages}{47--66}.
\newblock


\bibitem[\protect\citeauthoryear{Aum\"{u}ller, Christiani, Pagh, and
  Silvestri}{Aum\"{u}ller et~al\mbox{.}}{2018}]%
        {Aumuller18}
\bibfield{author}{\bibinfo{person}{Martin Aum\"{u}ller},
  \bibinfo{person}{Tobias Christiani}, \bibinfo{person}{Rasmus Pagh}, {and}
  \bibinfo{person}{Francesco Silvestri}.} \bibinfo{year}{2018}\natexlab{}.
\newblock \showarticletitle{Distance-Sensitive Hashing}. In
  \bibinfo{booktitle}{\emph{Proc. 37th ACM Symposium on Principles of Database
  Systems ({PODS})}}. \bibinfo{pages}{89--104}.
\newblock


\bibitem[\protect\citeauthoryear{Bar{-}Yossef, Jayram, Kumar, Sivakumar, and
  Trevisan}{Bar{-}Yossef et~al\mbox{.}}{2002}]%
        {BarYossefJKST02}
\bibfield{author}{\bibinfo{person}{Ziv Bar{-}Yossef}, \bibinfo{person}{T.~S.
  Jayram}, \bibinfo{person}{Ravi Kumar}, \bibinfo{person}{D. Sivakumar}, {and}
  \bibinfo{person}{Luca Trevisan}.} \bibinfo{year}{2002}\natexlab{}.
\newblock \showarticletitle{Counting Distinct Elements in a Data Stream}. In
  \bibinfo{booktitle}{\emph{Proc. 6th International Workshop Randomization and
  Approximation Techniques (RANDOM)}}. \bibinfo{pages}{1--10}.
\newblock


\bibitem[\protect\citeauthoryear{Broder}{Broder}{1997}]%
        {bro97b}
\bibfield{author}{\bibinfo{person}{Andrei~Z. Broder}.}
  \bibinfo{year}{1997}\natexlab{}.
\newblock \showarticletitle{On the resemblance and containment of documents}.
  In \bibinfo{booktitle}{\emph{Proc. Compression and Complexity of Sequences}}.
  \bibinfo{pages}{21--29}.
\newblock


\bibitem[\protect\citeauthoryear{Charikar}{Charikar}{2002}]%
        {Charikar02}
\bibfield{author}{\bibinfo{person}{Moses Charikar}.}
  \bibinfo{year}{2002}\natexlab{}.
\newblock \showarticletitle{Similarity estimation techniques from rounding
  algorithms}. In \bibinfo{booktitle}{\emph{Proc. 34th {ACM} Symposium on
  Theory of Computing (STOC)}}. \bibinfo{pages}{380--388}.
\newblock


\bibitem[\protect\citeauthoryear{Charikar and Siminelakis}{Charikar and
  Siminelakis}{2017}]%
        {CharikarS17}
\bibfield{author}{\bibinfo{person}{Moses Charikar} {and} \bibinfo{person}{Paris
  Siminelakis}.} \bibinfo{year}{2017}\natexlab{}.
\newblock \showarticletitle{Hashing-Based-Estimators for Kernel Density in High
  Dimensions}. In \bibinfo{booktitle}{\emph{Proc. 58th {IEEE} Annual Symposium
  on Foundations of Computer Science (FOCS)}}. \bibinfo{pages}{1032--1043}.
\newblock


\bibitem[\protect\citeauthoryear{Christiani}{Christiani}{2017}]%
        {christiani2017framework}
\bibfield{author}{\bibinfo{person}{Tobias Christiani}.}
  \bibinfo{year}{2017}\natexlab{}.
\newblock \showarticletitle{A Framework for Similarity Search with Space-Time
  Tradeoffs using Locality-Sensitive Filtering}. In
  \bibinfo{booktitle}{\emph{Proc. 28th Symposium on Discrete Algorithms
  (SODA)}}. \bibinfo{pages}{31--46}.
\newblock


\bibitem[\protect\citeauthoryear{David and Nagaraja}{David and
  Nagaraja}{2004}]%
        {david2004order}
\bibfield{author}{\bibinfo{person}{Herbert~Aron David} {and}
  \bibinfo{person}{Haikady~Navada Nagaraja}.} \bibinfo{year}{2004}\natexlab{}.
\newblock \showarticletitle{Order statistics}.
\newblock \bibinfo{journal}{\emph{Encyclopedia of Statistical Sciences}}
  (\bibinfo{year}{2004}).
\newblock


\bibitem[\protect\citeauthoryear{Dwork, Hardt, Pitassi, Reingold, and
  Zemel}{Dwork et~al\mbox{.}}{2012}]%
        {DworkHPRZ12}
\bibfield{author}{\bibinfo{person}{Cynthia Dwork}, \bibinfo{person}{Moritz
  Hardt}, \bibinfo{person}{Toniann Pitassi}, \bibinfo{person}{Omer Reingold},
  {and} \bibinfo{person}{Richard~S. Zemel}.} \bibinfo{year}{2012}\natexlab{}.
\newblock \showarticletitle{Fairness through awareness}. In
  \bibinfo{booktitle}{\emph{Proc. Innovations in Theoretical Computer Science
  (ITCS)}}. \bibinfo{pages}{214--226}.
\newblock


\bibitem[\protect\citeauthoryear{Eshghi and Rajaram}{Eshghi and
  Rajaram}{2008}]%
        {eshghi2008locality}
\bibfield{author}{\bibinfo{person}{Kave Eshghi} {and}
  \bibinfo{person}{Shyamsundar Rajaram}.} \bibinfo{year}{2008}\natexlab{}.
\newblock \showarticletitle{Locality sensitive hash functions based on
  concomitant rank order statistics}. In \bibinfo{booktitle}{\emph{Proc. 14th
  ACM SIGKDD international conference on Knowledge discovery and data mining
  (KDD)}}. ACM, \bibinfo{pages}{221--229}.
\newblock


\bibitem[\protect\citeauthoryear{Flajolet and Martin}{Flajolet and
  Martin}{1985}]%
        {FlajoletM85}
\bibfield{author}{\bibinfo{person}{Philippe Flajolet} {and}
  \bibinfo{person}{G.~Nigel Martin}.} \bibinfo{year}{1985}\natexlab{}.
\newblock \showarticletitle{Probabilistic Counting Algorithms for Data Base
  Applications}.
\newblock \bibinfo{journal}{\emph{J. Comput. Syst. Sci.}} \bibinfo{volume}{31},
  \bibinfo{number}{2} (\bibinfo{year}{1985}), \bibinfo{pages}{182--209}.
\newblock


\bibitem[\protect\citeauthoryear{Friedler, Scheidegger, and
  Venkatasubramanian}{Friedler et~al\mbox{.}}{2016}]%
        {FriedlerSV16}
\bibfield{author}{\bibinfo{person}{Sorelle~A. Friedler},
  \bibinfo{person}{Carlos Scheidegger}, {and} \bibinfo{person}{Suresh
  Venkatasubramanian}.} \bibinfo{year}{2016}\natexlab{}.
\newblock \showarticletitle{On the (im)possibility of fairness}.
\newblock \bibinfo{journal}{\emph{CoRR}}  \bibinfo{volume}{abs/1609.07236}
  (\bibinfo{year}{2016}).
\newblock


\bibitem[\protect\citeauthoryear{Har-Peled and Mahabadi}{Har-Peled and
  Mahabadi}{2019}]%
        {HarPeledM19}
\bibfield{author}{\bibinfo{person}{Sariel Har-Peled} {and}
  \bibinfo{person}{Sepideh Mahabadi}.} \bibinfo{year}{2019}\natexlab{}.
\newblock \showarticletitle{Near Neighbor: Who is the Fairest of Them All?}
\newblock In \bibinfo{booktitle}{\emph{Proc. Advances in Neural Information
  Processing Systems (NeurIPS)}}. \bibinfo{pages}{13176--13187}.
\newblock


\bibitem[\protect\citeauthoryear{Hardt, Price, and Srebro}{Hardt
  et~al\mbox{.}}{2016}]%
        {HardtPNS16}
\bibfield{author}{\bibinfo{person}{Moritz Hardt}, \bibinfo{person}{Eric Price},
  {and} \bibinfo{person}{Nati Srebro}.} \bibinfo{year}{2016}\natexlab{}.
\newblock \showarticletitle{Equality of Opportunity in Supervised Learning}. In
  \bibinfo{booktitle}{\emph{Proc. Annual Conference on Neural Information
  Processing Systems (NIPS)}}. \bibinfo{pages}{3315--3323}.
\newblock


\bibitem[\protect\citeauthoryear{Hu, Qiao, and Tao}{Hu et~al\mbox{.}}{2014}]%
        {HuQT14}
\bibfield{author}{\bibinfo{person}{Xiaocheng Hu}, \bibinfo{person}{Miao Qiao},
  {and} \bibinfo{person}{Yufei Tao}.} \bibinfo{year}{2014}\natexlab{}.
\newblock \showarticletitle{Independent range sampling}. In
  \bibinfo{booktitle}{\emph{Proc. 33rd {ACM} Symposium on Principles of
  Database Systems (PODS)}}. \bibinfo{pages}{246--255}.
\newblock


\bibitem[\protect\citeauthoryear{Indyk and Motwani}{Indyk and Motwani}{1998}]%
        {IndykM98}
\bibfield{author}{\bibinfo{person}{Piotr Indyk} {and} \bibinfo{person}{Rajeev
  Motwani}.} \bibinfo{year}{1998}\natexlab{}.
\newblock \showarticletitle{Approximate Nearest Neighbors: {T}owards Removing
  the Curse of Dimensionality}. In \bibinfo{booktitle}{\emph{Proc. 30th {ACM}
  Symposium on the Theory of Computing (STOC)}}. \bibinfo{pages}{604--613}.
\newblock


\bibitem[\protect\citeauthoryear{Kannan, Roth, and Ziani}{Kannan
  et~al\mbox{.}}{2019}]%
        {Kannan}
\bibfield{author}{\bibinfo{person}{Sampath Kannan}, \bibinfo{person}{Aaron
  Roth}, {and} \bibinfo{person}{Juba Ziani}.} \bibinfo{year}{2019}\natexlab{}.
\newblock \showarticletitle{Downstream Effects of Affirmative Action}. In
  \bibinfo{booktitle}{\emph{Proc. Conference on Fairness, Accountability, and
  Transparency (FAT*)}}.
\newblock


\bibitem[\protect\citeauthoryear{Knuth}{Knuth}{1997}]%
        {Knuth97}
\bibfield{author}{\bibinfo{person}{Donald~E. Knuth}.}
  \bibinfo{year}{1997}\natexlab{}.
\newblock \bibinfo{booktitle}{\emph{The Art of Computer Programming, Volume 2:
  Seminumerical Algorithms}}.
\newblock \bibinfo{publisher}{Addison-Wesley}.
\newblock


\bibitem[\protect\citeauthoryear{Koren, Bell, and Volinsky}{Koren
  et~al\mbox{.}}{2009}]%
        {KorenBV09}
\bibfield{author}{\bibinfo{person}{Yehuda Koren}, \bibinfo{person}{Robert~M.
  Bell}, {and} \bibinfo{person}{Chris Volinsky}.}
  \bibinfo{year}{2009}\natexlab{}.
\newblock \showarticletitle{Matrix Factorization Techniques for Recommender
  Systems}.
\newblock \bibinfo{journal}{\emph{{IEEE} Computer}} \bibinfo{volume}{42},
  \bibinfo{number}{8} (\bibinfo{year}{2009}), \bibinfo{pages}{30--37}.
\newblock


\bibitem[\protect\citeauthoryear{Leonhardt, Anand, and Khosla}{Leonhardt
  et~al\mbox{.}}{2018}]%
        {LeonhardtAK18}
\bibfield{author}{\bibinfo{person}{Jurek Leonhardt}, \bibinfo{person}{Avishek
  Anand}, {and} \bibinfo{person}{Megha Khosla}.}
  \bibinfo{year}{2018}\natexlab{}.
\newblock \showarticletitle{User Fairness in Recommender Systems}. In
  \bibinfo{booktitle}{\emph{Companion Proceedings of the The Web Conference
  (WWW)}}. \bibinfo{pages}{101--102}.
\newblock


\bibitem[\protect\citeauthoryear{Li and K{\"{o}}nig}{Li and
  K{\"{o}}nig}{2010}]%
        {LiK10}
\bibfield{author}{\bibinfo{person}{Ping Li} {and}
  \bibinfo{person}{Arnd~Christian K{\"{o}}nig}.}
  \bibinfo{year}{2010}\natexlab{}.
\newblock \showarticletitle{b-Bit minwise hashing}. In
  \bibinfo{booktitle}{\emph{Proc. of International Conference on World Wide Web
  (WWW)}}. \bibinfo{pages}{671--680}.
\newblock


\bibitem[\protect\citeauthoryear{Luong, Ruggieri, and Turini}{Luong
  et~al\mbox{.}}{2011}]%
        {LuongRT11}
\bibfield{author}{\bibinfo{person}{Binh~Thanh Luong},
  \bibinfo{person}{Salvatore Ruggieri}, {and} \bibinfo{person}{Franco Turini}.}
  \bibinfo{year}{2011}\natexlab{}.
\newblock \showarticletitle{k-NN As an Implementation of Situation Testing for
  Discrimination Discovery and Prevention}. In \bibinfo{booktitle}{\emph{Proc.
  17th ACM SIGKDD International Conference on Knowledge Discovery and Data
  Mining (KDD)}}. \bibinfo{pages}{502--510}.
\newblock


\bibitem[\protect\citeauthoryear{of~the President}{of~the President}{2016}]%
        {Whi16}
\bibfield{author}{\bibinfo{person}{Executive~Office of~the President}.}
  \bibinfo{year}{2016}\natexlab{}.
\newblock \bibinfo{title}{Big data: A report on algorithmic systems,
  opportunity, and civil rights.}
\newblock
\newblock


\bibitem[\protect\citeauthoryear{Olken and Rotem}{Olken and Rotem}{1995a}]%
        {Olken1995Survey}
\bibfield{author}{\bibinfo{person}{Frank Olken} {and} \bibinfo{person}{Doron
  Rotem}.} \bibinfo{year}{1995}\natexlab{a}.
\newblock \showarticletitle{Random sampling from databases: a survey}.
\newblock \bibinfo{journal}{\emph{Statistics and Computing}}
  \bibinfo{volume}{5}, \bibinfo{number}{1} (\bibinfo{year}{1995}),
  \bibinfo{pages}{25--42}.
\newblock


\bibitem[\protect\citeauthoryear{Olken and Rotem}{Olken and Rotem}{1995b}]%
        {Olken1995}
\bibfield{author}{\bibinfo{person}{Frank Olken} {and} \bibinfo{person}{Doron
  Rotem}.} \bibinfo{year}{1995}\natexlab{b}.
\newblock \showarticletitle{Sampling from spatial databases}.
\newblock \bibinfo{journal}{\emph{Statistics and Computing}}
  (\bibinfo{year}{1995}), \bibinfo{pages}{43--57}.
\newblock


\bibitem[\protect\citeauthoryear{Riazi, Chen, Shrivastava, Wallach, and
  Koushanfar}{Riazi et~al\mbox{.}}{2016}]%
        {Riazi16}
\bibfield{author}{\bibinfo{person}{M.~Sadegh Riazi}, \bibinfo{person}{Beidi
  Chen}, \bibinfo{person}{Anshumali Shrivastava}, \bibinfo{person}{Dan~S.
  Wallach}, {and} \bibinfo{person}{Farinaz Koushanfar}.}
  \bibinfo{year}{2016}\natexlab{}.
\newblock \bibinfo{title}{{Sub-Linear Privacy-Preserving Near-Neighbor Search
  with Untrusted Server on Large-Scale Datasets}}.  (\bibinfo{year}{2016}).
\newblock
\newblock
\shownote{ArXiv:1612.01835.}


\bibitem[\protect\citeauthoryear{Selbst, boyd, Friedler, Venkatasubramanian,
  and Vertesi}{Selbst et~al\mbox{.}}{2019}]%
        {Selbst19}
\bibfield{author}{\bibinfo{person}{Andrew~D. Selbst}, \bibinfo{person}{danah
  boyd}, \bibinfo{person}{Sorelle Friedler}, \bibinfo{person}{Suresh
  Venkatasubramanian}, {and} \bibinfo{person}{Janet Vertesi}.}
  \bibinfo{year}{2019}\natexlab{}.
\newblock \showarticletitle{Fairness and Abstraction in Sociotechnical
  Systems}. In \bibinfo{booktitle}{\emph{Proc. ACM Conference on Fairness,
  Accountability, and Transparency (FAT*)}}.
\newblock


\bibitem[\protect\citeauthoryear{Szarek and Werner}{Szarek and Werner}{1999}]%
        {SZAREK1999193}
\bibfield{author}{\bibinfo{person}{Stanislaw~J. Szarek} {and}
  \bibinfo{person}{Elisabeth Werner}.} \bibinfo{year}{1999}\natexlab{}.
\newblock \showarticletitle{A Nonsymmetric Correlation Inequality for Gaussian
  Measure}.
\newblock \bibinfo{journal}{\emph{Journal of Multivariate Analysis}}
  \bibinfo{volume}{68}, \bibinfo{number}{2} (\bibinfo{year}{1999}),
  \bibinfo{pages}{193 -- 211}.
\newblock
\showISSN{0047-259X}


\bibitem[\protect\citeauthoryear{Venkatasubramanian}{Venkatasubramanian}{2019}]%
        {Venkatasubramanian19}
\bibfield{author}{\bibinfo{person}{Suresh Venkatasubramanian}.}
  \bibinfo{year}{2019}\natexlab{}.
\newblock \showarticletitle{Algorithmic Fairness: Measures, Methods and
  Representations}. In \bibinfo{booktitle}{\emph{Proc. 38th {ACM} Symposium on
  Principles of Database Systems {(PODS)}}}. \bibinfo{pages}{481}.
\newblock


\bibitem[\protect\citeauthoryear{Williams}{Williams}{2005}]%
        {Williams05}
\bibfield{author}{\bibinfo{person}{Ryan Williams}.}
  \bibinfo{year}{2005}\natexlab{}.
\newblock \showarticletitle{A New Algorithm for Optimal 2-constraint
  Satisfaction and Its Implications}.
\newblock \bibinfo{journal}{\emph{Theor. Comput. Sci.}} \bibinfo{volume}{348},
  \bibinfo{number}{2} (\bibinfo{year}{2005}), \bibinfo{pages}{357--365}.
\newblock


\end{thebibliography}
\appendix

\section{Sampling with a repeated query}\label{app:nns_repeated} 
If we repeat the same query $\q$ several times (and no other queries), the data structure in Section \ref{sec:sampling} always returns the same point. 
If we let $OUT_i$ denote the output at the $i$-th repetition of query  $\q$, we have that
$\PR{OUT_i = \p | OUT_1=\p_1}$ is 1 if $\p=\p_1$ and 0 otherwise.
We would like to extend the above data structure to get
\begin{align}\label{eq:nnsq1}
&\PR{OUT_i = \p | OUT_{i-1}=\p_{i-1}, \ldots OUT_{1}=\p_{1}}= 
\\&= \PR{OUT_i = \p} = \frac{1}{b_S(\q,r)}\nonumber
\end{align}
for a given query $\q$ and $1<i\leq \BT{n}$ (i.e., there is a linear number of queries), and thus to solve the $r$-NNIS problem in the case when only one query is repeated.
The main idea is to select a near neighbor $\p$ of $\q$ as in the  data structure of Theorem~\ref{thm:nns} and, just before returning $\p$, to apply a small random perturbation to ranks for ``destroying'' any relevant information that can be collected by repeating the query.
The perturbation is obtained by applying a random swap, similar to the one in the Fisher-Yates  shuffle \cite{Knuth97}: we random select a point $\p'\in S$ (in fact, a suitable subset of $S$) and exchange the  ranks of $\p$ and $\p'$. 
In this way, we get a data structure satisfying Equation~\eqref{eq:nnsq1}
with  expected query time $\BO{(n^{\rho} + b_S(\q,cr)-b_S(\q,r) ) \log^2 n}$. 

We remark that the rerandomization technique is only restricted to single element queries: 
over time all elements in the ball $B_S(\q, r)$ get higher and higher ranks. This means that for another query $\q'$ such that $B_S(\q, r)$ and $B_S(\q', r)$ have a non-empty intersection, the elements in $B_S(\q', r) \setminus B_S(\q,r)$ become more and more likely to be returned. The next section will provide a slightly more involved data structure that guarantees independence among queries.

We first explain why two simple adaptations of the data structure for $r$-NNS in section \ref{sec:sampling} don't work with a repeated query. 
The first adaptation consists in returning the point with the $k$-th smallest rank in $S_\q = \bigcup_{i=1}^{L} S_{i,\ell_i(\q)}$ where $k$ is an integer in $[1,S_\q]$ randomly and independently selected at query time.
Although this approach gives Equations \ref{eq:nnsq1}, the query time is $\BO{n^\rho b_S(\q,r)+ b_S(\q,cr)}$ with constant probability, as in the trivial solution described at the beginning.
Indeed, since there are point repetitions among buckets,  we have to scan almost all buckets $S_{i,\ell_i(\q)}$ for every $i$, getting the same asymptotic query time of the trivial LSH solution.
Another  approach is to random select an integer value $r$ in  $[r_m, r_M]$, where $r_m$ and $r_M$ are the smallest and largest ranks in $S_\q$, and then to return the point in $S_\q$ with the closest rank to $r$.
However, this approach doesn't guarantee independence among the outputs: Indeed, before the first query the initial permutation is unknown and thus all points  are equally likely to be returned; subsequent queries however reveal some details on the initial distribution, conditioning  subsequent queries and exposing which points are more likely to be reported.
We now describe our solution.

\paragraph*{Construction}
The construction is as in the one in section \ref{sec:sampling}, with the exception that points in each bucket $S_{i,j}$ are stored in a priority queue, where  ranks give the priority. 
We assume that each queue supports the extraction of the point with \emph{minimum} rank and the insertion, deletion, and update of a point given the (unique) rank.

\paragraph*{Query}
The algorithm starts as in the previous data structure by searching for the point $\x$ in  $B_S(\q,r) \cap \bigcup_{i=1}^{L} S_{i,\ell_i(\q)}$ with lowest rank.
However, just before returning the sampled point, we update the rank of $\x$ as follows. 
Let $r_\x$ be the rank of $\x$; the algorithm selects uniformly at random  a value $r$ in $\{r_\x,\ldots n\}$ and let $\y$ be the point in $S$ with rank $r$; then the algorithm swaps the rank of $\x$ and $\y$ and updates all data structures (i.e., the priority queues of buckets $S_{i,\ell_i(\x)}$ and  $S_{i,\ell_i(\y)}$ with $1\leq i \leq L$).

\begin{theorem}\label{thm:nnsr}
With high probability $1-1/n$, the above data structure solves the $r$-NNIS problem with one single repeated query: for each  query point $\q$ and for each point $\p\in S$, we have:
\begin{enumerate}
\item $\p$ is returned as near neighbor of $\q$ with probability $1/b_S(\q,r)$.
\item $\PR{OUT_i = \p | OUT_{i-1}=\p_{i-1}, \ldots OUT_{1}=\p_{1}} = {1}/{b_S(\q,r)}$ for each $1<i\leq n$.
\end{enumerate}
Furthermore, the data structure requires $\BT{ n^{1+\rho}\log n}$ space and the expected query time is $$\BO{(n^{\rho} + b_S(\q,cr)-b_S(\q,r)) \log^2 n}.$$
\end{theorem}
\begin{proof}
We claim that after every query $\q$, it is not possible  to distinguish how the  $B_S(\q,r)$ near neighbors of $\q$ are distributed among the ranks $\{r_\x,\ldots n\}$.
Before the rank shuffle, the position of $\x$ is known while the positions of points in $B_S(\q,r)/\{\x\}$ in $\{r_\x+1,\ldots n\}$ still remain unknown.
After swapping $\x$ with a random point with rank $\{r_\x, n\}$, $\x$ can be in each rank with  probability $1/(n-r_\x+1)$ and all distributions of points in $B_S(\q,r)$ in ranks $\{r_\x, n\}$ are thus equally likely. 
As a consequence, properties 1 and 2 in the statement follow. 

Before finding a $r$-near neighbor, the query algorithm finds, with high probability, $\BO{(b_S(\q,cr)-b_S(\q,r)+L)\log n}$ points with distance larger than $r$.  
The final rank shuffle requires $\BO{L \log n}$ time as we need to update the at most $2L$ buckets and priority queues with points $\x$ and $\y$. The theorem follows.\footnote{We observe that if \emph{all} points touched at query time are updated with the rank-swap approach, then we get $\BO{(1+b_S(\q,cr)/b_S(\q,r)) n^\rho }$ expected query time. 
However, this bound is better than the one stated in Theorem~\ref{thm:nnsr} only when $b_S(\q,cr)=\BO{b_S(\q,r)}$.} 
\end{proof}

\section{A linear space near-neighbor data structure}
\label{app:tableau}
We will split up the analysis of the data structure from Section~\ref{sec:tableau} into two parts. First, we describe and analyse a query algorithm that ignores the cost of storing and evaluating the $m$ random vectors. Next, we will describe and analyze the changes necessary to obtain an efficient query method as the one described in Section~\ref{sec:tableau}.

\subsection{Description of the Data Structure}

\paragraph*{Construction} To set up the data structure for a point set $S \subseteq \mathbb{R}^d$ of $n$ data points and two parameters $\beta < \alpha$, 
choose $m \geq 1$ random vectors $\a_1, \ldots, \a_m$ where each $\a = (a_1, \ldots, a_d) \sim \mathcal{N}(0, 1)^d$ is a vector of $d$ independent and identically  distributed standard normal Gaussians. 
For each $i \in\{1,\ldots,m\}$, let $L_i$ contain all data points $\x \in S$ such that $\ip{\a_i}{\x}$ is largest among all vectors $\a$. 

\paragraph*{Query} For a query point $\q \in \mathbb{R}^d$ and for a choice of $\varepsilon \in (0, 1)$ controlling the success probability of the query, define $f(\alpha, \varepsilon) = \sqrt{2(1-\alpha^2)\ln(1/\varepsilon)}$. 
Let $\Delta_\q$ be the largest inner product of $\q$ over all $\a$. 
Let $L'_1, \ldots, L'_K$ denote the lists associated with random vectors $\a$ satisfying $\ip{\a}{\q} \geq \alpha \Delta_q - f(\alpha,\varepsilon)$. 
Check all points in $L'_1, \ldots, L'_K$ and report the first point $\x$ such that $\ip{\q}{\x} \geq \beta$. 
If no such point exists, report $\perp$.

The proof of the theorem below will ignore the cost of evaluating $\a_1, \ldots, \a_m$. 
An efficient algorithm for evaluating these vectors is provided in Section~\ref{sec:tableau:efficient}.

\begin{theorem}
\label{thm:tableau}
Let $-1 < \beta < \alpha < 1$, $\varepsilon \in (0, 1)$, and $n \geq 1$. 
Let $\rho = \frac{(1-\alpha^2)(1-\beta^2)}{(1 - \alpha\beta)^2}$. 
There exists $m = m(n, \alpha, \beta)$ such that the data structure described above solves the $(\alpha, \beta)$-NN problem with probability at least $1 - \varepsilon$ 
using space $O(m + n)$ and expected query time $n^{\rho + o(1)}$. 
\end{theorem}

We split the proof up into multiple steps. First, we show that for every choice of $m$, inspecting the lists associated with those random vectors $\a$ such that their inner product with the query point $\q$ is at least the given query threshold guarantees to find a close point with probability at least $1-\varepsilon$. The next step is to show that the number of far points in these lists is $n^{\rho + o(1)}$ in expectation. 

\subsection{Analysis of Close Points}

\begin{lemma}
\label{lem:tableau:close}
Given $m$ and $\alpha$, let $\q$ and $\x$ such that $\ip{\q}{\x} = \alpha$. 
    Then we find $\x$ with probability at least $1- \varepsilon$ in the lists associated with vectors that have inner product at least $\alpha\Delta_\q - f(\alpha, \varepsilon)$ with $\q$. 
\end{lemma}

\begin{proof}
By spherical symmetry \cite{christiani2017framework}, we may assume that $\x = (1, 0, \ldots, 0)$ and $\q = (\alpha, \sqrt{1 - \alpha^2}, 0,$ $  \ldots, 0)$ . The probability of finding $\x$ when querying the data structure for $\q$ can be bounded as follows from below. Let $\Delta_\x$ be the largest inner product of $\x$ with vectors $\a$ and let $\Delta_\q$ be the largest inner product of $\q$ with these vectors. Given these thresholds, finding $\x$ is then equivalent to the statement that for the
vector $\a$ with
$\ip{\a}{\x} = \Delta_\x$ we have $\ip{\a}{\q} \geq \alpha \Delta_{\q} - f(\alpha,\varepsilon)$. We note that $\PR{\max\{\ip{\a}{\q}\} = \Delta} = 1 - \PR{\forall i: \ip{\a_i}{\q} < \Delta}$. 

Thus, we may lower bound the probability of finding $\x$ for arbitrary choices $\Delta_{\x}$ and $\Delta_\q$ as follows:
\begin{align}
    \label{eq:tableau:success}
\PR{\text{find \hspace{-0.1em} $\x$}} \hspace{-0.1em} & \hspace{-0.1em} \geq \hspace{-0.1em} \PR{\hspace{-0.1em} \ip{\a}{\q}\hspace{-0.1em} \geq \hspace{-0.1em} \alpha\Delta_{\q}{-} f(\alpha,\varepsilon) \hspace{-0.1em} \mid \hspace{-0.1em} \ip{\a}{\x} = \Delta_{\x}\hspace{-0.1em}\text{ and } \hspace{-0.1em}\ip{\a'}{\q} \hspace{-0.1em} = \hspace{-0.1em} \Delta_{\q}\hspace{-0.1em}}\notag\\ 
    &\quad - \PR{\forall i: \ip{\a_i}{\x} < \Delta_\x}  - \PR{\forall i: \ip{\a_i}{\q} < \Delta_\q}.
\end{align}
Here, we used that $\PR{A \cap B \cap C} = 1 - \PR{\overline{A} \cup \overline{B} \cup \overline{C}} \geq \PR{A} - \PR{\overline{B}} - \PR{\overline{C}}$. 
We will now obtain bounds for the three terms on the right-hand side of \eqref{eq:tableau:success} separately, but we first recall the following lemma from \cite{SZAREK1999193}:
\begin{lemma}[\cite{SZAREK1999193}]\label{lem:normal:bound}
Let $Z$ be a standard normal random variable. Then, for every $t \geq 0$, we have that 
\begin{align*}
\frac{1}{\sqrt{2\pi}}\frac{1}{t + 1}e^{-t^2/2} \leq \Pr(Z \geq t) \leq \frac{1}{\sqrt{\pi}}\frac{1}{t + 1}e^{-t^2/2}.
\end{align*}
\end{lemma}

\paragraph{Bounding the first term.} Since $\q = (\alpha, \sqrt{1- \alpha^2}, 0, \ldots, 0)$ and $\x=(1, 0, \ldots, 0)$, the condition $\ip{\a}{\x} = \Delta_\x$ means that the first component of $\a$ is $\Delta_\x$. 
Thus, we have to bound the probability that a standard normal random variable $Z$ satisfies the inequality $\alpha \Delta_\x + \sqrt{1 - \alpha^2} Z \geq \alpha\Delta_\q - f(\alpha, \varepsilon)$. 
Reordering terms, we get 
\begin{align*}
Z \geq \frac{\alpha\Delta_{\q} - f(\alpha,\varepsilon) - \alpha \Delta_\x}{\sqrt{1 - \alpha^2}}.
\end{align*}
Choose $\Delta_\q = \Delta_\x$. 
In this case, we bound the probability that $Z$ is larger than a negative value. 
By symmetry of the standard normal distribution and using Lemma~\ref{lem:normal:bound}, we may compute
\begin{align}
    \label{eq:tableau:success:2}
    \PR{Z \geq -\frac{f(\alpha, \varepsilon)}{\sqrt{1 - \alpha^2}}}
    &=1 - \PR{Z < -\frac{f(\alpha, \varepsilon)}{\sqrt{1 - \alpha^2}}}
    \notag
    \\&=1 - \PR{Z \geq \frac{f(\alpha, \varepsilon)}{\sqrt{1 - \alpha^2}}}\notag\\
    &\geq 
    1 - \frac{\text{Exp}\left(-\frac{(f(\alpha, \varepsilon))^2}{2(1 - \alpha^2)}\right)}{\sqrt{2\pi}
            \left(\frac{f(\alpha, \varepsilon)}{\sqrt{1 - \alpha^2}} + 1\right) }
    \geq 1 - \varepsilon.
\end{align}

\paragraph{Bounding the second term and third term.} We first observe that 
\begin{align*}
    \PR{\forall i: \ip{\a_i}{\x} < \Delta_\x} &= \PR{\ip{\a_1}{\x} < \Delta_\x}^m \\
    &= \left(1 - \PR{\ip{\a_1}{\x} \geq \Delta_\x}\right)^m\\
        &\leq \left(1 - \frac{\text{Exp}[-\Delta_\x^2/2]}{\sqrt{2\pi}(\Delta_x + 1)}\right)^m.
\end{align*}
Setting $\Delta_\x = \sqrt{2 \log m - \log(4\kappa\pi \log(m))}$ upper bounds this term by $\text{Exp}[-\sqrt{\kappa}]$.
Thus, by setting $\kappa \geq \log^2(1/\delta)$ the second term is upper bounded by $\delta \in (0, 1)$. 
The same thought can be applied to the third summand of \eqref{eq:tableau:success}, which is only smaller because of the negative offset $f(\alpha, \varepsilon)$.



\paragraph{Putting everything together.} 
Putting the bounds obtained for all three summands together shows that we can find $\x$ with probability at least $1 - \varepsilon'$ by choosing $\varepsilon$ and $\delta$ such that $\varepsilon' \geq \varepsilon + 2 \delta$. 
\end{proof}

\subsection{Analysis of Far Points}

\begin{lemma}
    \label{lem:tableau:far}
    Let $-1 < \beta < \alpha < 1$. 
    There exists $m = m(n, \alpha, \beta)$ such that the expected number of points $\x$ with $\ip{\x}{\q} \leq \beta$ in $L'_1,\ldots,L'_K$ where $K =
    |\{i \mid \ip{\a_i}{\q} \geq \alpha\Delta_\q - f(\alpha, \varepsilon)\}|$ is $n^{\rho + o(1)}$.
\end{lemma}

\begin{proof}
    We will first focus on a single far-away point $\x$ with inner product at most $\beta$. 
    Again, let $\Delta_\q$ be the largest inner product of $\q$. Let $\x$ be stored in $L_i$. 
    Then we find $\x$ if and only if $\ip{\a_i}{\q} \geq \alpha\Delta_\q - f(\alpha, \varepsilon).$ 
    By spherical symmetry, we may assume that $\x = (1, 0, \ldots, 0)$ and $\q = (\beta, \sqrt{1 - \beta^2}, 0, \ldots, 0)$. 

    We first derive values $t_\q$ and $t_\x$ such that, with high probability, $\Delta_\q \geq t_\q$ and  $\Delta_\x \leq t_\x$. 
    From the proof of Lemma~\ref{lem:tableau:close}, we know that 
    \begin{align*}
        \PR{\max\{\ip{\a}{\q}\} \geq t}
        &\leq 1 - \left(1 - \frac{\text{Exp}\left(-t^2/2\right)}{\sqrt{2\pi}(t + 1)}\right)^m.
    \end{align*}

    Setting $t_\q = \sqrt{2\log(m / \log(n)) - \log(4 \pi \log(m
    / \log n))}$ shows that with high probability we have $\Delta_\q \geq t_\q$. Similarily, the choice $t_\x = \sqrt{2\log(m  \log(n)) - \log(4 \pi \log(m
     \log n))}$ is with high probability at least $\Delta_\x$. In the following, we condition on the event that $\Delta_\q \geq t_\q$ and $\Delta_\x \leq t_\x$.

     We may bound the probability of finding $\x$ as follows:
     \begin{align*}
     \PR{\ip{\a}{\q} \geq \alpha\Delta_\q - f(\alpha, \varepsilon) \mid \ip{\a}{\x} = \Delta_\x} \leq \\
     \leq \PR{\ip{\a}{\q} \geq \alpha\Delta_\q - f(\alpha, \varepsilon) \mid \ip{\a}{\x} = t_\x}\\
     \leq \PR{\ip{\a}{\q} \geq \alpha t_\q - f(\alpha, \varepsilon) \mid \ip{\a}{\x} = t_\x}.
     \end{align*}

    Given that $\ip{\a}{\x}$ is $t_\x$, the condition $\ip{\a}{\q} \geq
    \alpha t_\q - f(\alpha, \varepsilon)$ is equivalent to the statement that for a standard normal
    variable $Z$ we have $Z \geq \frac{(\alpha t_\q - f(\alpha, \varepsilon) - \beta t_\x)}{\sqrt{1 -
    \beta^2}}$. 
    Using Lemma~\ref{lem:normal:bound}, we have
    \begin{align}
       & \PR{\ip{\a}{\q} \geq \alpha t_\q 
        \hspace{-0.1em}- \hspace{-0.1em} f(\alpha, \varepsilon) \hspace{-0.2em} \mid \hspace{-0.2em} \ip{\a}{\x} = t_\x}
       \hspace{-0.2em} \leq \hspace{-0.2em} \frac{\text{Exp}\left(-\frac{(\alpha t_\q- f(\alpha, \varepsilon)- \beta t_\x)^2}{2(1 - \beta^2)}\right)}{\sqrt{\pi} \left(\frac{(\alpha t_\q- f(\alpha, \varepsilon)- \beta t_\x)}{\sqrt{1 - \beta^2}} + 1\right)}\notag\\
        &\hspace{2em}\leq  \text{Exp}\left(-\frac{(\alpha t_\q- f(\alpha, \varepsilon)- \beta t_\x)^2}{2(1 - \beta^2)}\right)\notag\\
        &\hspace{2em}\stackrel{(1)}{=} \text{Exp}\left(-\frac{(\alpha - \beta)^2 t_\x^2}{2(1 - \beta^2)} \left(1 + O(1 / \log \log n)\right)\right)\notag\\
        &\hspace{2em}= \left(\frac{1}{m}\right)^{\frac{(\alpha - \beta)^2}{1 - \beta^2} + o(1)},
        \label{eq:far:prob}
    \end{align}
    where step (1) follows from the observation that $t_\q/t_\x = 1 + O(1/\log \log n)$ and $f(\alpha, \varepsilon)/t_\x = O(1/\log\log n)$ if $m = \Omega(\log n)$.

    Next, we want to balance this probability with the expected cost for checking all lists where the inner product with the associated vector $\a$ is at least $\alpha \Delta_\q - f(\alpha, \varepsilon)$. 
    By linearity of expectation, the expected number of  lists to be checked is
   not more than $$m \cdot \text{Exp}\left(-(\alpha t_\q)^2\left(1/2 - f(\alpha,\varepsilon)/(\alpha t_\q) + f(\alpha, \varepsilon)^2/(2(\alpha t_\q)^2)\right)\right),$$
   which is $m^{1-\alpha^2 + o(1)}$ using the value of $t_\q$ set above.
    This motivates to set~\eqref{eq:far:prob} equal to $m^{1-\alpha^2} / n$, taking into account that there are at most $n$ far-away points. 
    Solving for $m$, we get 
    $
        m = n^{\frac{1 - \beta^2}{(1 - \alpha\beta)^2} + o(1)}
    $ and this yields $m^{1-\alpha^2 +o(1)} = n^{\rho + o(1)}$.
\end{proof}

\subsection{Efficient Evaluation}
\label{sec:tableau:efficient}
The previous subsections assumed that  
 $m$ filters can be evaluated and stored for free.
 However, this requires space and time $n^{(1-\beta^2)/(1-\alpha\beta)^2}$, which is much higher than the work we expect from checking the points in all filters above the threshold. We solve this problem by using the tensoring approach, which can be seen as a simplified version of the general approach proposed in~\cite{christiani2017framework}. 

 \paragraph*{Construction} Let $t = \lceil 1/(1 - \alpha^2)\rceil$ and assume that $m^{1/t}$ is an integer.
Consider $t$ independent data structures DS$_1$, $\ldots$, DS$_t$, each using $ m^{1/t}$ random vectors $\a_{i, j}$, for $i \in\{1,\ldots,t\}, j \in [m^{1/t}]$. 
Each DS$_i$ is instantiated as described above. 
During preprocessing, consider each $\x \in S$. 
If $\a_{1,i_1},\ldots,\a_{t,i_t}$ are the random vectors that achieve the largest inner product with $\x$ in DS$_1$, $\ldots$, DS$_t$, map the index of $\x$ in $S$ to the bucket $(i_1,\ldots,i_t) \in [m^{1/t}]^t$. 
Use a hash table to keep track of all non-empty buckets. 
Since each data point in $S$ is stored exactly once, the space usage is $O(n + tm^{1/t})$.

\paragraph*{Query} Given the query point $\q$, evaluate all $t m^{1/t}$ filters. 
For $i \in \{1, \ldots, t\}$, let $\mathcal{I}_i = \{j \mid \ip{\a_{i,j}}{\q} \geq \alpha \Delta_{\q, i} - f(\alpha, \varepsilon)\}$ be the set of all indices of filters that are above the individual query threshold in DS$_i$. 
Check all buckets $(i_1, \ldots, i_t) \in \mathcal{I}_1 \times \dots \times \mathcal{I}_t$. 
If there is a bucket containing a close point, return it, otherwise return $\perp$.

\begin{theorem}
Let $S \subseteq X$ with $|S| = n$ and $-1 < \beta < \alpha < 1$. The tensoring data structure solves the $(\alpha, \beta)$-NN problem in linear space and expected time $n^{\rho + o(1)}$.
\end{theorem}

Before proving the theorem, we remark that efficient evaluation comes at the price of lowering the success probability from a constant $p$ to $p^{1/(1-\alpha^2)}$. 
Thus, for $\delta \in (0,1)$ repeating the construction $\ln(1/\delta)p^{1-\alpha^2}$ times yields a success probability of at least $1-\delta$.

\begin{proof}
Observe that with the choice of $m$ as in the proof of Lemma~\ref{lem:tableau:far}, we can bound $m^{1/t} = n^{(1-\alpha^2)(1-\beta^2)/(1-\alpha\beta)^2 + o(1)} = n^{\rho + o(1)}$. 
This means that preprocessing takes time $n^{1+\rho + o(1)}$. 
Moreover, the additional space needed for storing the $t m^{1/t}$ random vectors is $n^{\rho + o(1)}$ as well. 
For a given query point $\q$, we expect that each $\mathcal{I}_i$ is of size $m^{(1-\alpha^2)/t + o(1)}$. 
Thus, we expect to check not more than $m^{1-\alpha^2 + o(1)}=n^{\rho+o(1)}$ buckets in the hash table, which shows the stated claim about the expected running time.

Let $\x$ be a point with $\ip{\q}{\x} \geq \alpha$. 
The probability of finding $\x$ is the probability that the vector associated with $\x$ has inner product at least $\alpha\Delta_{\q,i} - f(\alpha, \varepsilon)$ in DS$_i$, for all $i \in \{1, \ldots, t\}$. 
This probability is $p^t$, where $p$ is the probability of finding $\x$ in a single data structure $\text{DS}_i$. 
By Theorem~\ref{thm:tableau} and since $\alpha$ is a constant, this probability is constant and can be bounded from below by $1 - \delta$ via a proper choice of $\varepsilon$ as discussed in the proof of Lemma~\ref{lem:tableau:close}.

Let $\y$ be a point with $\ip{\q}{\y} < \beta$. 
Using the same approach in the proof of Lemma~\ref{lem:tableau:far}, we observe that the probability of finding $\y$ in an individual DS$_i$ is $(1/m)^{1/t \cdot (\alpha - \beta)^2/(1-\beta^2) + o(1)}$.
Thus the probability of finding $\y$ in a bucket inspected for $\q$ is at most $(1/m)^{(\alpha - \beta)^2/(1-\beta^2) + o(1)}$. Setting parameters as before shows that we expect at most $n^{\rho +o(1)}$ far points in buckets inspected for query $\q$, which completes the proof.
\end{proof}

\end{document}